\PassOptionsToPackage{round}{natbib}
\PassOptionsToPackage{pagebackref=true,hypertexnames=false}{hyperref}

\let\savedBibliographyStyle\bibliographystyle
\renewcommand{\bibliographystyle}[1]{}
\documentclass[10pt,twoside]{predoc2}
\let\bibliographystyle\savedBibliographyStyle
\hypersetup{allcolors=black}
\usepackage{mathtools,mathrsfs,bm,graphicx,placeins,float,etoolbox}
\input{ursfs.fd}
\DeclareFontShape{U}{rsfs}{m}{n}{%
  <5> <6> rsfs5
  <5.5> sfixed*[5]rsfs5
  <7> rsfs7
  <8> <9> <10> <10.95> <12> <14.4> <17.28> <20.74> <24.88> rsfs10
}{}
\usepackage[export]{adjustbox}
\usepackage[round]{natbib}
\usepackage{multibib}
\DeclareRobustCommand{\doi}[1]{%
  doi:\space\begingroup\urlstyle{rm}%
  \href{https://doi.org/#1}{\nolinkurl{#1}}\endgroup
}
\usepackage[capitalize,nameinlink]{cleveref}

\renewcommand*{\backref}[1]{}
\renewcommand*{\backrefalt}[4]{\ifnum#1>0 (#2)\fi}

\DeclareRobustCommand{\closedformsecref}{\ref{sec:closed-form-broja}}

\newtheoremstyle{defnoparen}%
  {3pt}{3pt}%
  {\normalfont}{}%
  {\bfseries}{.}%
  {.5em}%
  {\thmname{#1}\thmnumber{ #2}\thmnote{ \normalfont #3}}
\newtheoremstyle{plainnoparen}%
  {3pt}{3pt}%
  {\itshape}{}%
  {\bfseries}{.}%
  {.5em}%
  {\thmname{#1}\thmnumber{ #2}\thmnote{ \normalfont #3}}

\theoremstyle{defnoparen}
\newtheorem{definition}{Definition}
\newtheorem{assumption}{Assumption}
\newtheorem{remark}{Remark}

\theoremstyle{plainnoparen}
\newtheorem{lemma}{Lemma}
\newtheorem{proposition}{Proposition}
\newtheorem{theorem}{Theorem}
\newtheorem{corollary}{Corollary}

\AddToHook{env/definition/begin}{\crefalias{definition}{definition}}
\AddToHook{env/assumption/begin}{\crefalias{assumption}{assumption}}
\AddToHook{env/remark/begin}{\crefalias{remark}{remark}}
\AddToHook{env/lemma/begin}{\crefalias{lemma}{lemma}}
\AddToHook{env/proposition/begin}{\crefalias{proposition}{proposition}}
\AddToHook{env/theorem/begin}{\crefalias{theorem}{theorem}}
\AddToHook{env/corollary/begin}{\crefalias{corollary}{corollary}}

\newcommand{\Cov}{\operatorname{Cov}}

\newcommand{\BROJA}{\mathrm{BROJA}}
\newcommand{\UI}{\operatorname{UI}}
\newcommand{\RI}{\operatorname{RI}}
\newcommand{\SI}{\operatorname{SI}}

\usepackage{booktabs,multirow,threeparttable}
\newcites{fmri}{SUPPLEMENTAL BIBLIOGRAPHY}

\crefname{assumption}{Assumption}{Assumptions}
\Crefname{assumption}{Assumption}{Assumptions}
\newsavebox{\comparisonTableBox}
\newif\ifcomparisonTable
\AtBeginEnvironment{table}{\comparisonTabletrue}
\BeforeBeginEnvironment{tabular}{\ifcomparisonTable\begingroup\lrbox{\comparisonTableBox}\fi}
\AfterEndEnvironment{tabular}{\ifcomparisonTable\endlrbox\endgroup\ifdim\wd\comparisonTableBox>\linewidth\resizebox{\linewidth}{!}{\usebox{\comparisonTableBox}}\else\usebox{\comparisonTableBox}\fi\fi}
\AtBeginDocument{\setlength{\abovedisplayskip}{4pt plus 1pt minus 1pt}\setlength{\belowdisplayskip}{4pt plus 1pt minus 1pt}\setlength{\abovedisplayshortskip}{2pt}\setlength{\belowdisplayshortskip}{3pt}\setlength{\jot}{2pt}}

\newsavebox{\preprintFigureBox}
\NewDocumentCommand{\preprintNaturalGraphic}{O{} m}{%
  \sbox{\preprintFigureBox}{\preprintOriginalGraphics[#1]{#2}}%
  \ifdim\wd\preprintFigureBox>\linewidth
    \makebox[\linewidth][c]{\usebox{\preprintFigureBox}}%
  \else\usebox{\preprintFigureBox}\fi
}
\AtBeginEnvironment{figure}{\let\preprintOriginalGraphics\includegraphics\let\includegraphics\preprintNaturalGraphic}
\usepackage{wrapfig}

\usepackage{noindentafter}
\NoIndentAfterEnv{equation}
\NoIndentAfterEnv{equation*}
\NoIndentAfterEnv{align}
\NoIndentAfterEnv{align*}
\NoIndentAfterEnv{figure}
\NoIndentAfterEnv{figure*}
\NoIndentAfterEnv{table}
\NoIndentAfterEnv{table*}
\expandafter\NoIndentAfterCmd\csname]\endcsname

\makeatletter
\DeclareRobustCommand{\preprintAffiliationMark}[1]{%
  {\m@th\ensuremath{^{\mbox{\fontsize\sf@size\sf@size\selectfont\normalfont#1}}}}%
}
\renewcommand{\AB@authnote}[1]{\preprintAffiliationMark{#1}}
\renewcommand{\AB@affilnote}[1]{\preprintAffiliationMark{#1}}
\xpatchcmd{\affil}{\if\relax#1\relax}{\if\relax\detokenize{#1}\relax}{}{}
\makeatother
\title{Gaussian BROJA Partial Information Decomposition in Closed Form}
\author[2,3]{Ohad Sheelo}
\author[1,2,3,\textdagger]{Tal Golan}
\author[1,2,3,\textdagger]{Ari Pakman}
\date{}
\affil[1]{Department of Industrial Engineering and Management, Ben-Gurion University of the Negev, Israel}
\affil[2]{School of Brain Sciences and Cognition, Ben-Gurion University of the Negev, Israel}
\affil[3]{Data Science Research Center, Ben-Gurion University of the Negev, Israel}
\affil[\textdagger]{These authors jointly supervised this work.}
\makeatletter
\def\runauthor{Ohad Sheelo, Tal Golan, and Ari Pakman}
\def\runtitle{Gaussian BROJA Partial Information Decomposition in Closed Form}
\makeatother
\begin{document}
\addtocontents{toc}{\protect\setcounter{tocdepth}{-1}}
\maketitle
\begin{abstract}
  The BROJA partial information decomposition provides an operationally motivated way to separate the information that two sources carry about a target into unique, redundant, and synergistic components. For continuous and high-dimensional variables, however, computing the BROJA decomposition requires solving a difficult optimization problem. Previous work showed that, when each source--target pair is jointly Gaussian, the problem can be reduced to finite-dimensional numerical optimization. Here, we derive a closed-form solution to this optimization problem, enabling direct computation of the BROJA decomposition in the general multivariate Gaussian setting. The resulting closed-form expression enables speedups exceeding two orders of magnitude for high-dimensional Gaussian BROJA PID. Code is available at \url{https://github.com/brainsandmachines/Eigen_PID}.
\end{abstract}
\section{INTRODUCTION}

\begin{wrapfigure}[23]{r}{3.23in}
  \vspace{-10pt}
  \includegraphics{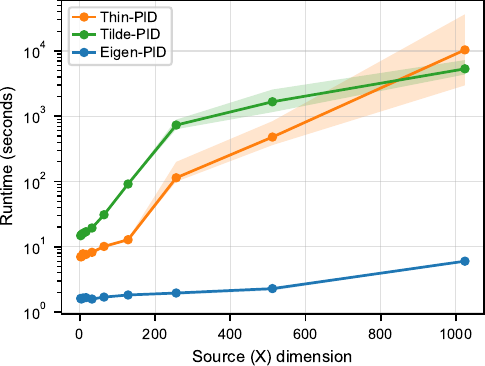}
  \caption{Runtime comparison of Tilde-PID \citep{venkatesh2023}, Thin-PID \citep{zhao2025} and Eigen-PID (ours), varying the dimensionality of source \(X\) ($d_X = 2^k$ for $k=1\ldots 10$) while holding the dimensions of source \(Y\) and target \(M\) fixed at~1000. Experimental details appear in \cref{subsec:supp-runtime-comparisons}. Lines show the median runtime across repetitions, and shaded regions indicate the 25th--75th percentile range.}%
  \label{fig:runtime_speedup}
\end{wrapfigure}

Many questions in science and engineering concern how several variables jointly inform about a quantity of interest: whether two inputs carry the same information about an output, whether each contributes something the other lacks, or whether the information emerges only from their combination. Mutual information alone cannot distinguish these situations. Partial information decomposition (PID), introduced by \citet{williamsbeer2010}, addresses this by splitting the total information that a set of sources provides about a target into redundant, unique, and synergistic parts. Originally developed for the analysis of neural and complex systems, PID has recently become a tool in machine learning: it has been used to quantify and model the interactions between modalities in multimodal learning \citep{liang2023quantifying}, to characterize the complexity of the representations learned by deep networks \citep{ehrlich2023measure}, to analyze how synergy supports multi-task learning in neural networks \citep{proca2024synergistic}, to define feature relevance and guide feature selection \citep{wollstadt2023rigorous,westphal2025partial}, and to dissect the sources of unfairness in federated learning \citep{hamman2024demystifying}. A central practical obstacle in these applications is the computation of PID for continuous, high-dimensional variables, which motivates the present work.

In this paper we show that, for the widely used BROJA decomposition of \citet{bertschinger2014quantifying}, this obstacle disappears in the Gaussian case: the BROJA components admit a closed-form solution in terms of the conditional covariances of the target given each source. The resulting procedure, which we call Eigen-PID, requires no iterative optimization and matches existing Gaussian PID methods at a small fraction of their runtime.

The paper is organized as follows.~\cref{sec:background} reviews PID and the BROJA optimization,~\cref{sec:preliminaries} introduces the tools used in our analysis, and~\cref{sec:closed-form-broja} derives the closed-form solution. \cref{sec:related-work} discusses how this result relates to and extends prior work, and~\cref{sec:numerical} presents numerical comparisons with existing methods. \cref{sec:supp-identities,sec:supp-preliminaries,sec:supp-section3,sec:supp-numerical,sec:fmri} contain standard identities, auxiliary lemmas, additional proofs, additional experiments, and an application to real neuroimaging data, respectively.

\section{BACKGROUND}
\label{sec:background}
Formally, in the bivariate setting, PID refines the information that two sources \(X\) and \(Y\) carry about a target \(M\) into four components satisfying
\begin{equation}
  \label{eq:pid-decomposition}
  I(M;X,Y) = \operatorname{RI} + \operatorname{UI}_X + \operatorname{UI}_Y + \operatorname{SI},
\end{equation}
together with
\begin{equation}
  \label{eq:pid-marginals}
  I(M;X) = \operatorname{RI} + \operatorname{UI}_X, \qquad I(M;Y) = \operatorname{RI} + \operatorname{UI}_Y.
\end{equation}
Here, \(\operatorname{RI}\) denotes the redundant information about \(M\), available from either source; \(\operatorname{UI}_X\) and \(\operatorname{UI}_Y\) denote the information available uniquely from \(X\) and from \(Y\), respectively; and \(\operatorname{SI}\) denotes the synergistic information, available only from the sources jointly.

The information contained in the unique and redundant components, excluding synergy, is referred to as the union information \citep{griffith2014quantifying}:
\begin{equation}
  I_{\cup}(M;X,Y)
  :=
  \operatorname{RI}
  +
  \operatorname{UI}_X
  +
  \operatorname{UI}_Y.
\end{equation}
Equivalently,
\begin{equation}
  I(M;X,Y)
  =
  I_{\cup}(M;X,Y)
  +
  \operatorname{SI}.
\end{equation}
Thus, the union information represents the portion of the total information about \(M\) that is not synergistic.

A central approach to defining the individual PID components was proposed by \citet{bertschinger2014quantifying}, commonly referred to as the BROJA PID. Their construction rests on an assumption, labeled~\((*)\) in their paper, which states that the unique and redundant information depend only on the two source--target marginals \(P_{MX}\) and \(P_{MY}\), rather than on the full joint distribution \(P_{MXY}\). Accordingly, given two prescribed source--target marginals with a common target marginal \(P_M\), one considers the following set of joint distributions.

\begin{definition}[Compatible couplings]
  \label{def:compatible-couplings}
  The set of joint distributions compatible with the two prescribed target--source marginals is
  \begin{equation}
    \Delta(P)=
    \{Q_{MXY}:Q_{MX}\!=P_{MX},\ Q_{MY}\!=P_{MY}\}.
    \label{eq:compatible-couplings}
  \end{equation}
  Each $Q\in\Delta(P)$ is called a compatible coupling: it specifies the dependence between $X$ and $Y$ while preserving both target--source marginals.
\end{definition}

In particular, \(I_Q(M;X)=I(M;X)\) and \(I_Q(M;Y)=I(M;Y)\) for every \(Q\in\Delta(P)\), so that, under assumption~\((*)\), the redundant and unique information are invariant across \(\Delta(P)\), whereas the synergistic information \(\operatorname{SI}_Q\) may change with the coupling between the sources. Since, for every \(Q\in\Delta(P)\),
\begin{align}
  I_Q(M;X\mid Y)
  \label{eq:conditional-mutual-info}
   & =
  I_Q(M;X,Y)-I(M;Y)
  \\
   & =
  \operatorname{UI}_X+\operatorname{SI}_Q,
\end{align}
minimizing the conditional mutual information over \(\Delta(P)\) amounts to minimizing the synergistic component over the admissible couplings between the sources. This motivates the following definitions.

\begin{definition}[BROJA unique and union information]
  \label{def:broja-info}
  The BROJA unique information carried by $X$ relative to $Y$ \citep{bertschinger2014quantifying} is
  \begin{equation}
    \operatorname{UI}^{\mathrm{BROJA}}_X
    :=
    \inf_{Q\in\Delta(P)} I_Q(M;X\mid Y).
    \label{eq:broja-unique}
  \end{equation}
  The associated marginal-constrained union information \citep{griffith2014quantifying} is
  \begin{equation}
    I_{\cup}^{\mathrm{BROJA}}(M;X,Y)
    :=
    \inf_{Q\in\Delta(P)} I_Q(M;X,Y).
    \label{eq:broja-union}
  \end{equation}
\end{definition}

The unique information of \(Y\) is defined analogously, interchanging \(X\) and \(Y\). Combining (\ref{eq:conditional-mutual-info}), (\ref{eq:broja-unique}) and (\ref{eq:broja-union}), we obtain that the BROJA unique and union informations satisfy, whenever $I(M;Y)<\infty$,
\begin{equation}
  \operatorname{UI}^{\mathrm{BROJA}}_X
  =
  I_{\cup}^{\mathrm{BROJA}}(M;X,Y)-I(M;Y).
  \label{eq:ui-union-relation}
\end{equation}
Evaluating the BROJA unique-information components is therefore equivalent to evaluating the union-information minimum, which is the optimization solved in this paper.

In the Gaussian setting with \(d_M=1\), \citet{barrett2015} showed that the BROJA PID reduces to the minimum-mutual-information decomposition,
\begin{equation}
  \mathrm{RI}^{\mathrm{BROJA}}
  =
  \min\{I(M;X),I(M;Y)\},
\end{equation}
and therefore
\begin{equation}
  \begin{aligned}
    \mathrm{UI}^{\mathrm{BROJA}}_X
     & =
    [I(M;X)-I(M;Y)]_+, \\
    \mathrm{UI}^{\mathrm{BROJA}}_Y
     & =
    [I(M;Y)-I(M;X)]_+,
  \end{aligned}
\end{equation}
where \([a]_+:=\max\{a,0\}\). Hence, at least one of the two sources necessarily has zero unique information. For \(d_M>1\), however, no closed-form expression for the BROJA minimum was previously known, and for continuous, high-dimensional variables, optimizing over the full space \(\Delta(P)\) is computationally difficult. To address the Gaussian setting, \citet{venkatesh2023} considered jointly Gaussian marginals \(P_{MX}\) and \(P_{MY}\) and restricted the feasible coupling itself to be jointly Gaussian. Define
\begin{equation*}
  \Delta_G(P)
  :=
  \left\{
  Q\in\Delta(P):
  Q_{MXY}\ \text{is jointly Gaussian}
  \right\}.
\end{equation*}
The restricted union-information problem
\begin{equation}
  \inf_{Q\in\Delta_G(P)} I_Q(M;X,Y),
\end{equation}
can be reduced to a finite-dimensional optimization over the Gaussian coupling between the sources and evaluated using a projected-gradient algorithm. Since \(\Delta_G(P)\subseteq\Delta(P)\), this restriction initially provided only
\begin{equation}
  I_{\cup}^{\mathrm{BROJA}}(M;X,Y)
  \le
  \inf_{Q\in\Delta_G(P)} I_Q(M;X,Y),
\end{equation}
and it remained unknown whether restricting the optimization to jointly Gaussian couplings could increase the optimum.

\citet{zhao2025} subsequently resolved this question. They proved that when the two prescribed source-target marginals \(P_{MX}\) and \(P_{MY}\) are jointly Gaussian, a jointly Gaussian optimizer exists for the unrestricted problem. Consequently,
\begin{equation}
  \inf_{Q\in\Delta(P)}
  I_Q(M;X,Y)
  =
  \inf_{Q\in\Delta_G(P)}
  I_Q(M;X,Y).
\end{equation}
They further introduced Thin-PID, a more efficient gradient-based method for solving the resulting Gaussian optimization. Therefore, the unrestricted BROJA problem with Gaussian source--target marginals is known to reduce exactly to a finite-dimensional Gaussian optimization. Nevertheless, existing methods still require iterative numerical optimization. The remaining question is whether this population minimum admits a closed-form solution.

\subsection*{Main contribution}
To state our result, let \(\Sigma_M\) denote the covariance of \(M\), and let \(\Gamma_X\) and \(\Gamma_Y\) denote the conditional covariances of \(M\) given \(X\) and given \(Y\), respectively. Under the Gaussian assumption these are constant matrices, given explicitly in \cref{sec:preliminaries}. Let us define the eigendecomposition
\begin{align}
   & \Gamma_X^{-1/2}\Gamma_Y\Gamma_X^{-1/2}
  =
  U\Lambda U^\top,
  \\
   & \text{with } \,\, \Lambda=\operatorname{diag}(\lambda_1,\ldots,\lambda_{d_M}),
\end{align}
and \(U\) orthogonal, and define
\begin{equation}
  \begin{aligned}
    \Gamma^\star
     & :=
    \Gamma_X^{1/2}\,U D\,U^\top\Gamma_X^{1/2}, \\
    D
     & :=
    \operatorname{diag}\bigl(\min\{1,\lambda_1\},\ldots,\min\{1,\lambda_{d_M}\}\bigr).
  \end{aligned}
\end{equation}
In this paper, we prove that the Gaussian BROJA union information \eqref{eq:broja-union} admits the closed form
\begin{equation}
  I_{\cup}^{\mathrm{BROJA}}(M;X,Y)
  =
  \frac{1}{2}
  \log_2
  \frac{\det\Sigma_M}{\det\Gamma^\star},
\end{equation}
and, consequently, that the BROJA unique information components \eqref{eq:broja-unique} are
\begin{equation}
  \begin{aligned}
    \mathrm{UI}^{\mathrm{BROJA}}_X
     & =
    \frac{1}{2}
    \sum_{i=1}^{d_M}[\log_2\lambda_i]_+, \\
    \mathrm{UI}^{\mathrm{BROJA}}_Y
     & =
    \frac{1}{2}
    \sum_{i=1}^{d_M}[-\log_2\lambda_i]_+.
  \end{aligned}
\end{equation}
Compared with iterative numerical methods, our closed-form result has two immediate advantages:
\begin{enumerate}
  \item \textbf{Optimality guarantee:} It evaluates the exact population minimum directly, without initialization, convergence, or stopping criteria.

  \item \textbf{Efficiency:} It substantially reduces computation time as the dimensionality increases. In our simulations the resulting speedup exceeds \(100\times\) in high-dimensional settings, while the PID values coincide with those of the existing numerical methods.
\end{enumerate}

\FloatBarrier

\newpage
\section{PRELIMINARIES}
\label{sec:preliminaries}

All logarithms in information quantities are base $2$. Let $M\in\mathbb{R}^{d_M}$ be the target, and let $X\in\mathbb{R}^{d_X}$ and $Y\in\mathbb{R}^{d_Y}$ be two sources. We are given two target--source distributions $P_{MX}$ and $P_{MY}$ with the same target marginal $P_M$.

\subsection{Gaussian Setup}
\label{subsec:broja-gaussian-setup}
{\bf Notation:} For symmetric matrices $A$ and $B$, $A\preceq B$ denotes the
Loewner order, meaning that $B-A$ is positive semidefinite, i.e.,
$u^\top(B-A)u\geq0$ for every vector $u$.
Similarly, $A\prec B$ denotes that $B-A$ is positive definite.

\begin{assumption}[Nondegenerate Gaussian target--source marginals]
  \label{ass:gaussian-marginals}
  The prescribed marginals $P_{MX}$ and $P_{MY}$ are nondegenerate jointly Gaussian distributions with a common target marginal
  \begin{equation}
    M\sim\mathcal{N}(\mu_M,\Sigma_M),
    \qquad
    \Sigma_M\succ0.
    \label{eq:gaussian-target}
  \end{equation}
  Here, nondegeneracy means that the joint covariance matrices $\Sigma_{(M,X)}$ and $\Sigma_{(M,Y)}$ are positive definite. No Gaussian assumption is imposed on an arbitrary compatible coupling $Q_{MXY}\in\Delta(P)$.
\end{assumption}

Define the conditional covariance matrices
\begin{equation}
  \Gamma_X:=\operatorname{Cov}(M\mid X),
  \qquad
  \Gamma_Y:=\operatorname{Cov}(M\mid Y),
  \label{eq:posterior-covariances}
\end{equation}
given by the Schur complements
\begin{equation}
  \begin{aligned}
    \Gamma_X
     & =
    \Sigma_M-\Sigma_{MX}\Sigma_X^{-1}\Sigma_{XM},
    \\
    \Gamma_Y
     & =
    \Sigma_M-\Sigma_{MY}\Sigma_Y^{-1}\Sigma_{YM}.
  \end{aligned}
  \label{eq:posterior-schur}
\end{equation}
The formulas are standard; see \citet{anderson2003}. Nondegeneracy (\cref{ass:gaussian-marginals}), and the strict form of~\eqref{eq:schur-criterion-R} in \ref{supp:definition} imply the left inequalities below,
\begin{equation}
  0\prec\Gamma_X\preceq\Sigma_M,
  \qquad
  0\prec\Gamma_Y\preceq\Sigma_M,
  \label{eq:posterior-order}
\end{equation}
and the right ones hold since $\Sigma_{MX}\Sigma_X^{-1}\Sigma_{XM}\succeq0$ and $\Sigma_{MY}\Sigma_Y^{-1}\Sigma_{YM}\succeq0$.

\subsection{Gaussian and Matrix Tools}
\label{subsec:gaussian-matrix-tools}

We recall here some standard matrix-analysis definitions and results. See, e.g., \citet{hornjohnson2012}.

\noindent\textbf{Generalized eigenvalues.}
Since $\Gamma_X$ and $\Gamma_Y$ are positive definite, their
generalized eigenvalues are the positive scalars
$\lambda_1,\ldots,\lambda_{d_M}$ satisfying
\begin{equation}
  \Gamma_Yv_i=\lambda_i\Gamma_Xv_i
  \label{eq:generalized-eigenvalues}
\end{equation}
for generalized eigenvectors $v_i \neq 0$. Equivalently, let
\begin{align}
   & \Gamma_X^{-1/2}\Gamma_Y\Gamma_X^{-1/2}
  \label{eq:symmetric-eigendecomposition}
  =
  U\Lambda U^\top,
  \\
   & \Lambda :=\operatorname{diag}(\lambda_1,\ldots,\lambda_{d_M}),
\end{align}
where $U$ is orthogonal. Indeed, with $w_i=\Gamma_X^{1/2}v_i$, Equation~\eqref{eq:generalized-eigenvalues} is equivalent to $\Gamma_X^{-1/2}\Gamma_Y\Gamma_X^{-1/2}w_i=\lambda_i w_i$, and the matrix $\Gamma_X^{-1/2}\Gamma_Y\Gamma_X^{-1/2}$ is symmetric positive definite, so it admits the orthonormal eigendecomposition~\eqref{eq:symmetric-eigendecomposition} with positive eigenvalues; $u_i = w_i/\|w_i\|$ are the columns of $U$. Define
\begin{equation}
  T:=U^\top\Gamma_X^{-1/2}.
  \label{eq:T-def}
\end{equation}
Then $T$ is invertible and simultaneously diagonalizes $\Gamma_X$ and $\Gamma_Y$ by congruence:
\begin{equation}
  T\Gamma_XT^\top=I_{d_M},
  \qquad
  T\Gamma_YT^\top=\Lambda.
  \label{eq:simultaneous-diagonalization}
\end{equation}

\noindent\textbf{Congruence lemma.}
For every invertible matrix $T$ and symmetric
matrices $A,B$ of
compatible dimensions,
\begin{equation}
  A\preceq B
  \quad\Longleftrightarrow\quad
  TAT^\top\preceq TBT^\top.
  \label{eq:congruence-order}
\end{equation}
Indeed,
\begin{equation}
  TBT^\top-TAT^\top=T(B-A)T^\top,
\end{equation}
and congruence by an invertible matrix preserves positive semidefiniteness.

\noindent\textbf{Inversion reverses the Loewner order.}
If $0\prec A\preceq B$, then $B^{-1}\preceq A^{-1}$. Indeed, by the
congruence lemma~\eqref{eq:congruence-order} with $T=B^{-1/2}$,
\begin{equation}
  A\preceq B
  \quad\Longleftrightarrow\quad
  C:=B^{-1/2}AB^{-1/2}\preceq I.
  \label{eq:inversion-step1}
\end{equation}
Since $C$ is symmetric positive definite, it admits an eigendecomposition
\begin{equation}
  C=V\operatorname{diag}(c_1,\ldots,c_d)V^\top,
  \qquad
  V^\top V=I,
\end{equation}
and $0\prec C\preceq I$ means $0<c_i\leq1$ for every $i$. Hence $c_i^{-1}\geq1$, and
\begin{equation}
  C^{-1}
  =
  V\operatorname{diag}(c_1^{-1},\ldots,c_d^{-1})V^\top
  \succeq
  I.
  \label{eq:inversion-step2}
\end{equation}
Finally, $C^{-1}=B^{1/2}A^{-1}B^{1/2}$, so applying the congruence lemma~\eqref{eq:congruence-order} with $T=B^{-1/2}$ to Equation~\eqref{eq:inversion-step2} gives
\begin{equation}
  A^{-1}
  =
  B^{-1/2}C^{-1}B^{-1/2}
  \succeq
  B^{-1/2} B^{-1/2}
  =
  B^{-1}.
\end{equation}

\noindent\textbf{Hadamard's determinant inequality.}
For every
positive-semidefinite matrix $C$,
\begin{equation}
  \det C\leq\prod_{i=1}^{d_M}C_{ii}.
  \label{eq:hadamard}
\end{equation}
If, moreover, $C\succ0$, equality holds if and only if $C$ is diagonal.

\noindent\textbf{Gaussian-channel representation.}
Following \citet{blackwell1953}, we call any random variable $A$ jointly
distributed with $M$, i.e., any Markov kernel $P_{A\mid M}$, an
\emph{experiment} about $M$.
For any pair of jointly Gaussian variables $(M,A)$,
after centering both variables, the variable $A$ can be represented
as a linear transformation of $M$ plus independent Gaussian noise:
\begin{equation}
  \begin{aligned}
    A   & =H_A M+N_A,             \\
    N_A & \sim\mathcal{N}(0,R_A),
    \qquad
    N_A\perp\!\!\!\perp M,
  \end{aligned}
  \label{eq:gaussian-experiment}
\end{equation}
where $H_A = \Sigma_{AM}\Sigma_M^{-1}$ and $R_A = \Sigma_{A|M} = \Sigma_{A} - H_A\Sigma_{MA}$. The posterior covariance $\Gamma_A =\Sigma_{M|A}$ satisfies the posterior-precision identity
\begin{equation}
  \Gamma_A^{-1}
  =
  \Sigma_M^{-1}+H_A^\top R_A^{-1}H_A,
  \label{eq:posterior-precision}
\end{equation}
which follows from applying the Woodbury matrix identity~\eqref{eq:woodbury} to the Schur complement formula~\eqref{eq:posterior-schur} for $\Gamma_A$ (\cref{sec:supp-identities}).

\section{CLOSED-FORM GAUSSIAN BROJA PID}
\label{sec:closed-form-broja}

We now use the preceding ingredients to evaluate the unrestricted Gaussian BROJA population-level optimization in closed form.

\subsection{Maximum-Determinant Optimizer}
\label{subsec:max-det-solution}

Define the feasible family of posterior-error covariance matrices by
\begin{equation}
  \mathcal{F}_{X,Y}
  :=
  \left\{
  \Gamma\in\mathbb{S}_{++}^{d_M}:
  \Gamma\preceq\Gamma_X,\
  \Gamma\preceq\Gamma_Y
  \right\},
  \label{eq:feasible-posterior-covariances}
\end{equation}
where $\mathbb{S}_{++}^{d_M}$ denotes the set of $d_M\times d_M$ symmetric positive-definite matrices. Let
\begin{equation}
  \Gamma^\star
  \in
  \operatorname*{arg\,max}_{\Gamma\in\mathcal{F}_{X,Y}}
  \det\Gamma.
  \label{eq:max-det-posterior-covariance}
\end{equation}
The existence and uniqueness of $\Gamma^\star$, together with its closed-form expression, are established next.

\begin{proposition}[Exact maximum-determinant optimizer]
  \label{prop:max-det}
  Let $\lambda_1,\ldots,\lambda_{d_M}$ and $T$ be defined by Equations~\eqref{eq:generalized-eigenvalues}--\eqref{eq:simultaneous-diagonalization}. Define
  \begin{equation}
    D:=\operatorname{diag}(d_1,\ldots,d_{d_M}),
    \qquad
    d_i:=\min\{1,\lambda_i\}.
    \label{eq:D-def}
  \end{equation}

  Recall that \begin{equation}
    T:=U^\top\Gamma_X^{-1/2}.
  \end{equation}
  Then the optimizer in Equation~\eqref{eq:max-det-posterior-covariance} exists uniquely and is given by
  \begin{equation}
    \Gamma^\star=T^{-1}DT^{-\top}.
    \label{eq:gamma-star}
  \end{equation}
  In particular,
  \begin{equation}
    \det\Gamma^\star
    =
    \det(T^{-1})^2
    \prod_{i=1}^{d_M}\min\{1,\lambda_i\}.
    \label{eq:gamma-star-det}
  \end{equation}
\end{proposition}

\begin{proof}
  We first verify feasibility. Since every $\lambda_i$ is positive, $D\succ0$. Moreover,
  \begin{equation}
    D\preceq I_{d_M},
    \qquad
    D\preceq\Lambda.
  \end{equation}
  Using Equations~\eqref{eq:simultaneous-diagonalization}, \eqref{eq:congruence-order} and~\eqref{eq:gamma-star},
  \begin{equation}
    \Gamma_X-\Gamma^\star
    =
    T^{-1}(I_{d_M}-D)T^{-\top}
    \succeq0
  \end{equation}
  and
  \begin{equation}
    \Gamma_Y-\Gamma^\star
    =
    T^{-1}(\Lambda-D)T^{-\top}
    \succeq0.
  \end{equation}
  Hence $\Gamma^\star\in\mathcal{F}_{X,Y}$.

  Now let $\Gamma\in\mathcal{F}_{X,Y}$ be arbitrary and define
  \begin{equation}
    C:=T\Gamma T^\top.
    \label{eq:C-def}
  \end{equation}
  By the congruence lemma~\eqref{eq:congruence-order},
  \begin{equation}
    C\succ0,
    \qquad
    C\preceq I_{d_M},
    \qquad
    C\preceq\Lambda.
  \end{equation}
  Therefore, for every $i$,
  \begin{equation}
    C_{ii}\leq1,
    \qquad
    C_{ii}\leq\lambda_i,
  \end{equation}
  and consequently
  \begin{equation}
    C_{ii}\leq d_i=\min\{1,\lambda_i\}.
  \end{equation}
  Hadamard's determinant inequality~\eqref{eq:hadamard} gives
  \begin{equation}
    \det C
    \leq
    \prod_{i=1}^{d_M}C_{ii}
    \leq
    \prod_{i=1}^{d_M}d_i
    =
    \det D.
    \label{eq:C_D_inequality}
  \end{equation}
  Taking determinants in Equation~\eqref{eq:C-def}, and applying the multiplicativity of the determinant, gives
  \begin{align}
    \det\Gamma
     & =
    \det(T^{-1})^2\det C
    \\
     & \leq
    \det(T^{-1})^2\det D
    =
    \det\Gamma^\star.
  \end{align}
  Thus, $\Gamma^\star$ is optimal.

  Finally, to prove uniqueness, assume that $\det C = \det D$. Then both inequalities in Equation~\eqref{eq:C_D_inequality} must be equalities. Equality in Hadamard's inequality~\eqref{eq:hadamard} requires the positive-definite matrix $C$ to be diagonal. Since $0<C_{ii}\leq d_i$ for every $i$, equality $\prod_i C_{ii}=\prod_i d_i$ forces $C_{ii}=d_i$ for every $i$: if $C_{jj}<d_j$ for some $j$, then
  \begin{equation}
    \det C = \prod_{i=1}^{d_M}C_{ii}
    \leq
    C_{jj}\prod_{i\neq j}d_i
    <
    \prod_{i=1}^{d_M}d_i = \det D,
  \end{equation}
  a contradiction. Hence $C=D$, and therefore $\Gamma=\Gamma^\star$ and the optimizer is unique.
\end{proof}

In these generalized-eigenvector coordinates, the optimizer simply keeps the smaller of the two posterior variances in each target direction.

\subsection{Lower Bound on the Union Information}
\label{subsec:lower-bound}

\begin{proposition}[BROJA union information lower bound]
  \label{prop:broja-lower-bound}
  Under \cref{ass:gaussian-marginals},
  \begin{equation}
    I_{\cup}^{\mathrm{BROJA}}(M;X,Y)
    \geq
    \frac12\log_2
    \frac{\det\Sigma_M}{\det\Gamma^\star}.
    \label{eq:broja-lower-bound}
  \end{equation}
  The proof is given in~\cref{supp:proof-broja-lower-bound}.
\end{proposition}

\subsection{Achievability and the Closed-Form Solution}
\label{subsec:achievability}

\begin{theorem}[Gaussian BROJA union and unique information]
  \label{thm:gaussian-broja}
  Under \cref{ass:gaussian-marginals}, let $\lambda_1,\ldots,\lambda_{d_M}$ be the generalized eigenvalues defined by
  \begin{equation}
    \Gamma_Yv_i=\lambda_i\Gamma_Xv_i,
  \end{equation}
  and let $\Gamma^\star$ be defined by Equation~\eqref{eq:gamma-star}. Then
  \begin{equation}
    I_{\cup}^{\mathrm{BROJA}}(M;X,Y)
    =
    \frac12\log_2
    \frac{\det\Sigma_M}{\det\Gamma^\star}.
    \label{eq:main-union}
  \end{equation}
  Consequently,
  \begin{equation}
    \operatorname{UI}^{\mathrm{BROJA}}_X
    =
    \frac12
    \sum_{i=1}^{d_M}[\log_2\lambda_i]_+,
    \label{eq:main-ui-x}
  \end{equation}
  and
  \begin{equation}
    \operatorname{UI}^{\mathrm{BROJA}}_Y
    =
    \frac12
    \sum_{i=1}^{d_M}[-\log_2\lambda_i]_+.
    \label{eq:main-ui-y}
  \end{equation}
  Moreover, the infimum in the BROJA optimization is attained by a compatible jointly Gaussian coupling.
\end{theorem}
\newpage
\begin{proof}
  We first construct a Gaussian experiment $Z^\star$ attaining the lower bound of \cref{prop:broja-lower-bound}. By \cref{prop:max-det},
  \begin{equation}
    0\prec\Gamma^\star\preceq\Gamma_X\preceq\Sigma_M.
  \end{equation}
  Since inversion reverses the Loewner order, the matrix
  \begin{equation}
    J^\star
    :=
      (\Gamma^\star)^{-1}-\Sigma_M^{-1}
    \label{eq:J-star}
  \end{equation}
  is positive semidefinite. Define the linear Gaussian experiment
  \begin{equation}
    \begin{aligned}
      Z^\star
       & :=
           (J^\star)^{1/2}(M-\mu_M)+\varepsilon, \\
      \varepsilon
       & \sim\mathcal{N}(0,I_{d_M}),
      \qquad
      \varepsilon\perp\!\!\!\perp M.
    \end{aligned}
    \label{eq:Z-star}
  \end{equation}
  Applying the posterior-precision identity \eqref{eq:posterior-precision} gives
  \begin{equation}
    \operatorname{Cov}(M\mid Z^\star)^{-1}
    =
    \Sigma_M^{-1}+J^\star
      =
      (\Gamma^\star)^{-1}.
  \end{equation}
  Therefore,
  \begin{equation}
    \operatorname{Cov}(M\mid Z^\star)
    =
    \Gamma^\star.
    \label{eq:Z-star-posterior}
  \end{equation}
  \cref{prop:max-det} also gives
  \begin{equation}
    \Gamma^\star\preceq\Gamma_X,
    \qquad
    \Gamma^\star\preceq\Gamma_Y.
  \end{equation}

  Since $(M,Z^\star)$ is jointly Gaussian,
\begin{align}
I(M;Z^\star)
& =
\frac12\log_2
\frac{\det\Sigma_M}
{\det\operatorname{Cov}(M| Z^\star)},
\\
& =
\frac12\log_2
\frac{\det\Sigma_M}{\det\Gamma^\star}.
\end{align}

  \vspace{0.6em}
  We next show that the optimum is attained by a compatible jointly
  Gaussian coupling $Q \in \Delta(P)$.
  Since $\operatorname{Cov}(M\mid Z^\star)=\Gamma^\star$
  by Equation~\eqref{eq:Z-star-posterior}, with
  $\Gamma^\star\preceq\Gamma_X$ and $\Gamma^\star\preceq\Gamma_Y$,
  \cref{lem:gaussian-garbling} (\cref{app:gaussian-garbling}), applied
  with $Z=Z^\star$ and $A=X$ and then with $A=Y$, provides
  linear-Gaussian garbling
  kernels $K_X(dx\mid z)$ and $K_Y(dy\mid z)$
  that generate copies $X^\star$ and $Y^\star$ of the source
  experiments from $Z^\star$. Generate $X^\star$ and $Y^\star$
  independently conditional on $Z^\star$. Then
  \begin{equation}
    (M,X^\star)\overset{d}{=}(M,X),
    \qquad
    (M,Y^\star)\overset{d}{=}(M,Y),
  \end{equation}
  so the joint law of $(M,X^\star,Y^\star)$ belongs to $\Delta(P)$. Moreover,
  \begin{equation}
    M-Z^\star-(X^\star,Y^\star),
  \end{equation}
  and therefore, from the data-processing inequality~\eqref{eq:data-processing},
  \begin{equation}
    I(M;X^\star,Y^\star)\leq I(M;Z^\star).
    \label{eq:XY_Z-star-data-processing}
  \end{equation}
  Combining now Equation \eqref{eq:broja-lower-bound} from \cref{prop:broja-lower-bound}, the definition of BROJA union information \eqref{eq:broja-union} and (\ref{eq:XY_Z-star-data-processing}) gives
  \begin{equation}
    \begin{aligned}
      \frac12\log_2
      \frac{\det\Sigma_M}{\det\Gamma^\star}
       & \leq
      I_{\cup}^{\mathrm{BROJA}}(M;X,Y)
      \\
       & \leq
      I(M;X^\star,Y^\star)
      \\
       & \leq
      I(M;Z^\star)
      =
      \frac12\log_2
      \frac{\det\Sigma_M}{\det\Gamma^\star}.
    \end{aligned}
  \end{equation}
  Every inequality is therefore an equality, proving that $(M,X^\star,Y^\star)$ is an attaining jointly Gaussian coupling, and proving Equation~\eqref{eq:main-union}.

  It remains to evaluate the two unique-information components. By Equation~\eqref{eq:ui-union-relation} and Equation~\eqref{eq:main-union},
  \begin{align}
    \operatorname{UI}^{\mathrm{BROJA}}_X
     & =
    \frac12\log_2
    \frac{\det\Sigma_M}{\det\Gamma^\star}
    -\frac12\log_2
    \frac{\det\Sigma_M}{\det\Gamma_Y}
    \\
     & =\frac12\log_2
    \frac{\det\Gamma_Y}{\det\Gamma^\star}.
    \label{eq:ui-x-det-ratio}
  \end{align}
  Using Equations~\eqref{eq:simultaneous-diagonalization} and~\eqref{eq:gamma-star},
  \begin{equation}
    \Gamma_Y=T^{-1}\Lambda T^{-\top},
    \qquad
    \Gamma^\star=T^{-1}DT^{-\top}.
  \end{equation}
  Consequently,
 \begin{equation}
\frac{\det\Gamma_Y}{\det\Gamma^\star}
=
\frac{\det\Lambda}{\det D}
=
\prod_{i=1}^{d_M}
\frac{\lambda_i}{\min\{1,\lambda_i\}}.
\end{equation}
For each $i$,
\begin{equation}
\log_2
\frac{\lambda_i}{\min\{1,\lambda_i\}}
=
[\log_2\lambda_i]_+.
\end{equation}
Substitution into Equation~\eqref{eq:ui-x-det-ratio} proves
Equation~\eqref{eq:main-ui-x}.

Finally, interchanging $X$ and $Y$ replaces the multiset of
generalized eigenvalues by
$\{\lambda_i^{-1}\}_{i=1}^{d_M}$. Hence
\begin{equation}
\operatorname{UI}^{\mathrm{BROJA}}_Y
=
\frac12
\sum_{i=1}^{d_M}[\log_2\lambda_i^{-1}]_+
=
\frac12
\sum_{i=1}^{d_M}[-\log_2\lambda_i]_+,
\end{equation}
which proves Equation~\eqref{eq:main-ui-y}.
\end{proof}

\begin{corollary}
  \label{supp:corollary 1}
  Using the result established in the proof of \cref{thm:gaussian-broja}, we can derive closed-form expressions for the BROJA redundancy and synergy. For a detailed derivation, see \cref{supp:proof-full-pid}.
\end{corollary}

\begin{remark}[Blackwell interpretation of the optimizer]
  \label{rem:blackwell}
  The proof generates the attaining coupling from the single Gaussian experiment $Z^\star$ through the Markov chains $M-Z^\star-X^\star$ and $M-Z^\star-Y^\star$, with $(M,X^\star)\overset{d}{=}(M,X)$ and $(M,Y^\star)\overset{d}{=}(M,Y)$. In the language of the comparison of experiments \citep{blackwell1953}, $Z^\star$ is thus at least as informative about $M$ as each source, i.e., a common refinement of $X$ and $Y$ in the sense of \citet{kolchinsky2022novel}. Among Gaussian experiments with this property, $Z^\star$ is the least informative about $M$: any such $Z$ yields, by the same construction, a compatible coupling $Q_Z\in\Delta(P)$ with $I(M;Z)\geq I_{Q_Z}(M;X,Y)\geq\frac12\log_2(\det\Sigma_M/\det\Gamma^\star)=I(M;Z^\star)$, by the data-processing inequality~\eqref{eq:data-processing} and \cref{prop:broja-lower-bound}.
\end{remark}

\section{RELATED WORK}
\label{sec:related-work}

\paragraph{Secret-key agreement.}
\citet{james2019} considered several secret-key agreement scenarios for quantifying unique information. In their \emph{elephant} scenario---named after the ancient Indian parable of blind men describing different parts of an elephant---the sources are viewed as different partial observations of a common target. To quantify the information about $M$ uniquely captured by $X$, the target may send a public message, available to both sources, while the observations of $M$ and $X$ are used to establish a shared secret that remains secret from $Y$. The corresponding one-way secret-key rate is $S(X\leftarrow M\Vert Y)$. James et al.\ drew a conceptual connection between BROJA unique information and this elephant secret-key scenario.

For discrete variables, \citet{banerjee2018} established a formal connection by showing that the corresponding one-way secret-key rate is upper-bounded by the BROJA unique information. In a work unrelated to PID, \citet{watanabe2011} derived the one-way secret-key rate for vector Gaussian sources with unrestricted public communication in closed form. In our notation,
\begin{equation}
  S(X\leftarrow M\Vert Y)
  =
  \frac12\sum_{i=1}^{d_M}[\log_2\lambda_i]_+ .
\end{equation}
Our result establishes that, in the Gaussian setting considered here, this one-way secret-key rate is the exact minimum of the BROJA unique-information optimization:
\begin{equation}
  S(X\leftarrow M\Vert Y)
  =
  \operatorname{UI}^{\mathrm{BROJA}}_X .
  \label{eq:broja-elephant-equivalence}
\end{equation}

\paragraph{Other analytical BROJA solutions.} \citet{goswami2024} derived analytical BROJA PID expressions for several structured families. For Gaussian distributions, their multivariate result allows vector-valued sources but is limited to a scalar target and assumes conditions yielding a compatible Markov coupling between the source experiments. Our result addresses the general multivariate Gaussian case, allowing the target and both sources to be multivariate, and does not require such a Markov coupling.

\newpage
\begin{figure}[!b]
  \centering \includegraphics{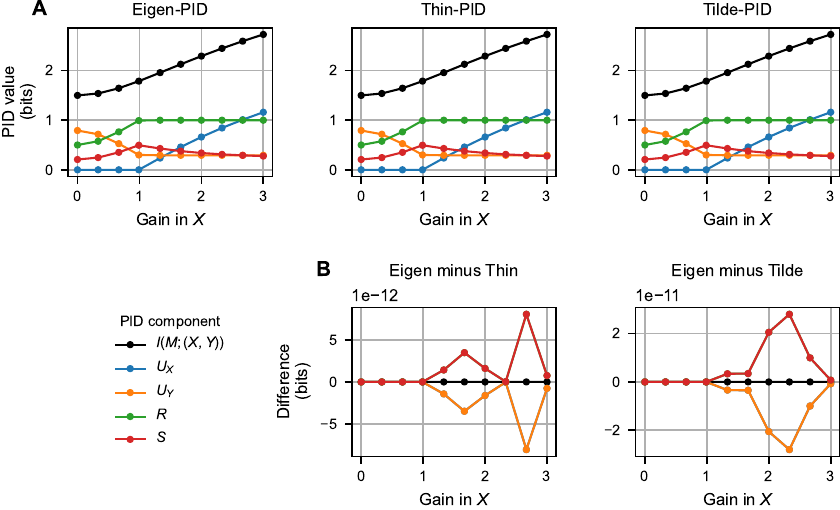}
  \caption{
    \textbf{Close numerical agreement between the closed-form solution and optimization-based methods.} Gain-sweep benchmark from Example~8 of \citet{venkatesh2023}, using the experimental setting of \citet{zhao2025}. The gain \(\alpha\) of \(X_1\) along target direction \(M_1\) is varied from \(0\) to \(3\), while the remaining channel parameters are fixed. For \(\alpha<1\), \(Y\) is more informative along both target directions; at \(\alpha=1\), \(X_1\) and \(Y_1\) are equally informative; and for \(\alpha>1\), \(X\) becomes more informative along \(M_1\), while \(Y\) remains more informative along \(M_2\). \textbf{(A)} PID components computed by Eigen-PID, Thin-PID, and Tilde-PID. \textbf{(B)} differences between Eigen-PID and the numerical methods. The differences remain below \(2.8\cdot10^{-11}\) bits for Tilde-PID and below \(8\cdot10^{-12}\) bits for Thin-PID. }
  \label{fig:gain_sweep}
\end{figure}

\section{IMPLEMENTATION AND NUMERICAL VALIDATION}
\label{sec:numerical}

We implemented the closed-form Eigen-PID solution in Python. The code, numerical examples, and usage instructions are available at \url{https://github.com/brainsandmachines/Eigen_PID}.

\paragraph{Numerical agreement.}

\Cref{fig:gain_sweep} reproduces the gain-sweep experiment introduced by \citet{venkatesh2023} and compares the resulting PID components across Eigen-PID, Thin-PID, and Tilde-PID (further details in \cref{subsec:supp-benchmark-examples}). For each PID component, the differences between the three methods are below \(2.8\cdot10^{-11}\) bits. To further verify optimality numerically, we initialized Tilde-PID and Thin-PID at an optimal Gaussian coupling derived from our solution. Absolute differences in union information were at most \(5.68\cdot10^{-11}\) bits across all trials, with mean absolute differences below \(4.1\cdot10^{-14}\) bits for both methods (\cref{fig:optimal-initialization}).

\paragraph{Runtime comparisons.}

We compared Eigen-PID, Thin-PID, and Tilde-PID on Gaussian models with one target and two sources (experimental setup in \cref{subsec:supp-runtime-comparisons}). We increased the dimension of one source while fixing the target and other source dimensions at 1000. All methods were evaluated on the same population covariance matrices using identical CPU resources. \Cref{fig:runtime_speedup} shows speedups of up to three orders of magnitude; \cref{fig:runtime_grid} extends the comparison to other target and source dimensions.

\clearpage
\section{DISCUSSION}
We have provided a closed-form solution for Gaussian BROJA PID. The solution shows that the unique information carried by each source is determined directly by the generalized eigenvalues of the two posterior covariance matrices and gives a closed-form expression for $\Gamma^\star$, the conditional covariance of $M$ given both sources under an optimal Gaussian coupling. We demonstrate that Eigen-PID agrees with Thin-PID and Tilde-PID up to numerical error while reducing runtime by orders of magnitude in high-dimensional settings.

\paragraph{Positive-definiteness requirements.} Applying Eigen-PID to data requires the estimated source--target joint covariance matrices to satisfy the positive-definiteness requirement in \cref{ass:gaussian-marginals}. For example, this requirement does not hold for sample covariance estimates when $N \leq d_M+\max(d_X,d_Y)$, where $N$ is the sample size. We recommend using regularization or dimensionality reduction to obtain positive-definite, well-conditioned covariance estimates. Once this positive-definiteness requirement is satisfied, the Schur complements defining $\Gamma_X$ and $\Gamma_Y$ are well defined and positive definite.

\paragraph{Correcting finite-sample bias.} As with existing Gaussian PID methods, replacing population covariances with sample covariances yields a plug-in estimator with finite-sample bias, even if the covariance estimates are well-conditioned. The resampling-based corrections of \citet{lorenz2026} can address this bias. Using Eigen-PID in place of their optimization-based estimator makes these corrections practical even for high-dimensional data. We provide an implementation of this approach and demonstrate its utility using real visual neuroimaging data (\cref{sec:fmri}). Our main fMRI PID analysis took 14.5 minutes using Eigen-PID, compared with a projected runtime of 21 days using Thin-PID (\cref{tab:fmri_runtime}). Moving the Eigen-PID computation from the CPU to a GPU further reduced the runtime to 3.9 minutes.

\paragraph{Extension to non-Gaussian data.} For non-Gaussian data, \citet{zhao2025} use invertible normalizing flows to map the variables into a latent space with approximately Gaussian source-target marginals. Gaussian PID is then computed in this latent space using Thin-PID. Eigen-PID can replace Thin-PID at this stage, removing the iterative optimization required for the Gaussian PID computation.

\paragraph{Differentiating through the PID computation.} Our PyTorch implementation of the closed-form PID is differentiable almost everywhere (with ReLU-like kinks at $\lambda_i=1$), enabling the exact Gaussian BROJA PID components to be used as loss terms in gradient-based optimization.

\clearpage
\subsection*{AI use statement}

In this work, we used generative AI tools to formulate mathematical claims, assist in the writing of proofs, and implement methods. Specifically, ChatGPT 5.6 Pro was used during exploratory mathematical work and generated preliminary versions of some proofs. The final versions of those proofs use substantially different arguments and do not rely on critical proof ingredients from the AI-generated versions. Codex assisted with software implementation from human-provided algorithmic specifications.

We did not use generative AI tools to develop theoretical models or conceptual frameworks, provide critical ingredients for the proofs ultimately presented, propose or refine hypotheses, design or provide feedback on research methodology or experiments, or interpret results. Generating synthetic datasets, assisting with translation, cleaning or reformatting datasets, and supporting qualitative or thematic data analysis are not applicable to this work.

Additionally, we used ChatGPT to generate copyediting and proofreading suggestions. AI tools also assisted with confirming consistency among descriptions of numerical results, figures, and experimental logs.

We reviewed all AI-assisted material. The mathematical claims and proofs in the paper and supplementary material were either developed by the authors or, where AI-assisted, manually verified and revised by the authors. AI-assisted code was manually verified and edited. We take responsibility for the final content of this work, including text, claims, proofs, code, and other artifacts produced with the aid of generative AI.

\subsubsection*{Acknowledgements}
This research was supported by the Israel Science Foundation (grants 534/24 to T.G. and 1138/23 to A.P.) and by the Israel Ministry of Innovation, Science and Technology (Israel–France Collaboration Program, 2025–2028, to A.P.).

\bibliographystyle{plainnat-doifirst}
\bibliography{references}

\clearpage
\appendix
\setcounter{figure}{0}\renewcommand{\thefigure}{S\arabic{figure}}
\setcounter{table}{0}
\renewcommand{\thetable}{S\arabic{table}}
\thispagestyle{empty}
\onecolumn

\begin{center}
  \makebox[\textwidth][c]{\Large\bfseries Gaussian BROJA Partial Information Decomposition in Closed Form:}\\[2pt]
  {\Large\bfseries Supplementary Material}
\end{center}

\addtocontents{toc}{\protect\setcounter{tocdepth}{2}}
\renewcommand{\contentsname}{Contents}
\tableofcontents

\section{DEFINITIONS AND USEFUL IDENTITIES}
\label{supp:definition}
\paragraph{Covariance}
Let $A \in \mathbb{R}^{d_A}$ and $B \in \mathbb{R}^{d_B}$ be random vectors with finite second moments. We denote their covariance matrices by
\[
  \Sigma_A
  :=
  \Cov(A)
  =
  \mathbb{E}\!\left[
    (A-\mathbb{E}[A])(A-\mathbb{E}[A])^\top
    \right]
  \in \mathbb{R}^{d_A \times d_A},
\]
and
\[
  \Sigma_B
  :=
  \Cov(B)
  =
  \mathbb{E}\!\left[
    (B-\mathbb{E}[B])(B-\mathbb{E}[B])^\top
    \right]
  \in \mathbb{R}^{d_B \times d_B}.
\]

Their cross-covariance matrices are
\[
  \Sigma_{AB}
  :=
  \Cov(A,B)
  =
  \mathbb{E}\!\left[
    (A-\mathbb{E}[A])(B-\mathbb{E}[B])^\top
    \right]
  \in \mathbb{R}^{d_A \times d_B},
\]
and
\[
  \Sigma_{BA}
  :=
  \Cov(B,A)
  =
  \mathbb{E}\!\left[
    (B-\mathbb{E}[B])(A-\mathbb{E}[A])^\top
    \right]
  \in \mathbb{R}^{d_B \times d_A},
\]
with
\[
  \Sigma_{BA}=\Sigma_{AB}^{\top}.
\]

We denote the joint covariance matrix of the stacked vector $(A,B)$ by
\[
  \Sigma_{(A,B)}
  :=
  \Cov
  \begin{pmatrix}
    A \\
    B
  \end{pmatrix}
  =
  \begin{pmatrix}
    \Sigma_A    & \Sigma_{AB} \\
    \Sigma_{BA} & \Sigma_B
  \end{pmatrix}
  \in
  \mathbb{R}^{(d_A+d_B)\times(d_A+d_B)}.
\]

\label[appendix]{sec:supp-identities}

We collect four standard facts used in the main text and in the proofs.

\paragraph{Differential entropy of a Gaussian vector.}
If $W\sim\mathcal{N}(\mu,\Sigma)$ in $\mathbb{R}^{d}$ with $\Sigma\succ0$, then
\begin{equation}
  h(W)
  =
  \frac12\log_2\bigl((2\pi e)^{d}\det\Sigma\bigr).
  \label{eq:gaussian-entropy}
\end{equation}
Consequently, if $(M,A)$ is jointly Gaussian with $M\sim\mathcal{N}(\mu_M,\Sigma_M)$ and constant conditional covariance $\Gamma_A=\operatorname{Cov}(M\mid A)\succ0$, then
\begin{equation}
  I(M;A)
  =
  h(M)-h(M\mid A)
  =
  \frac12\log_2\frac{\det\Sigma_M}{\det\Gamma_A}.
  \label{eq:gaussian-mi}
\end{equation}

\paragraph{Woodbury matrix identity.}
For matrices of compatible dimensions with $A$,$C$ and (A+UCV) invertible,
\begin{equation}
  (A+UCV)^{-1}
  =
  A^{-1}-A^{-1}U\bigl(C^{-1}+VA^{-1}U\bigr)^{-1}VA^{-1}.
  \label{eq:woodbury}
\end{equation}
Applied to the Schur complement~\eqref{eq:posterior-schur} with $A=\Sigma_M$, $U=-\Sigma_{MA}$, $C=\Sigma_A^{-1}$, and $V=\Sigma_{AM}$, it gives
\begin{equation}
  \Gamma_A^{-1}
  =
  \Sigma_M^{-1}
  +
  \Sigma_M^{-1}\Sigma_{MA}
  \bigl(\Sigma_A-\Sigma_{AM}\Sigma_M^{-1}\Sigma_{MA}\bigr)^{-1}
  \Sigma_{AM}\Sigma_M^{-1}
  =
  \Sigma_M^{-1}+H_A^\top R_A^{-1}H_A,
\end{equation}
with $H_A=\Sigma_{AM}\Sigma_M^{-1}$ and $R_A=\Sigma_A-\Sigma_{AM}\Sigma_M^{-1}\Sigma_{MA}$, which is the posterior-precision identity~\eqref{eq:posterior-precision}.

\paragraph{Data-processing inequality.}
If $M-Z-A$ is a Markov chain, i.e., $M$ and $A$ are conditionally independent given $Z$, then
\begin{equation}
  I(M;A)\leq I(M;Z).
  \label{eq:data-processing}
\end{equation}
In particular, this holds whenever $A$ is generated from $Z$ by a Markov kernel independent of $M$, as in the garbling of~\cref{lem:gaussian-garbling}.

\paragraph{Schur complement criterion.}
Let $P$ and $R$ be symmetric matrices and $Q$ a matrix of compatible dimensions. If $P\succ0$, then
\begin{equation}
  \begin{pmatrix} P & Q \\ Q^\top & R \end{pmatrix}\succeq0
  \quad\Longleftrightarrow\quad
  R-Q^\top P^{-1}Q\succeq0,
  \label{eq:schur-criterion}
\end{equation}
and, symmetrically, if $R\succ0$, then
\begin{equation}
  \begin{pmatrix} P & Q \\ Q^\top & R \end{pmatrix}\succeq0
  \quad\Longleftrightarrow\quad
  P-QR^{-1}Q^\top\succeq0.
  \label{eq:schur-criterion-R}
\end{equation}
Indeed, with the invertible matrix $T:=\bigl(\begin{smallmatrix} I & 0 \\ Q^\top P^{-1} & I \end{smallmatrix}\bigr)$,
\begin{equation}
  \begin{pmatrix} P & Q \\ Q^\top & R \end{pmatrix}
  =
  T
  \begin{pmatrix} P & 0 \\ 0 & R-Q^\top P^{-1}Q \end{pmatrix}
  T^\top,
\end{equation}
so by the congruence lemma~\eqref{eq:congruence-order} the block matrix is positive semidefinite if and only if the block-diagonal matrix is, i.e., since $P\succ0$, if and only if $R-Q^\top P^{-1}Q\succeq0$. The second statement follows by interchanging the roles of the two diagonal blocks. See, e.g., \citet[Section~7.7]{hornjohnson2012}. Both equivalences also hold with $\succ$ in place of $\succeq$, since congruence by an invertible matrix also preserves positive definiteness: for $u\neq0$, $u^\top TAT^\top u=(T^\top u)^\top A\,(T^\top u)$ with $T^\top u\neq0$.

\newpage
\section{AUXILIARY LEMMAS}
\label[appendix]{sec:supp-preliminaries}


\subsection{Posterior-Error Covariance Lemmas}
\label[appendix]{supp:posterior-error-lemmas}

For an arbitrary experiment $A$ about $M$, define its average posterior-error covariance by
\begin{align}
  \bar{\Gamma}_A
  := &
  \mathbb{E}_{P(M,A)} \Big[
                        (M-\mathbb{E}[M\mid A])
                        (M-\mathbb{E}[M\mid A])^\top
                        \Big]
  \label{eq:average-posterior-error}
  \\
  =  &
  \mathbb{E}_{P(A)}[\operatorname{Cov}(M\mid A)].
\end{align}
In general $\operatorname{Cov}(M\mid A)$ may depend on the observed value of $A$, so we use the barred notation $\bar{\Gamma}_A$ for its expectation. If $(M,A)$ is jointly Gaussian, then $\operatorname{Cov}(M\mid A)$ is constant almost surely. In that case,
\[
  \bar{\Gamma}_A=\operatorname{Cov}(M\mid A),
\]
and we use the unbarred notation $\Gamma_A$ for this common posterior covariance matrix, as in the main text. In particular,
\[
  \bar{\Gamma}_X=\Gamma_X,
  \qquad
  \bar{\Gamma}_Y=\Gamma_Y.
\]
The following lemma was proved by \citet[Proposition~5, Eq.~(23a)]{rioul2011} as an estimation-theoretic data-processing inequality.\footnote{In the notation of \citet[Section~I-A]{rioul2011}, the symbols $\leq$ and $\geq$ denote the Loewner order, and the conditional covariance is the averaged posterior-error covariance used here.}
\begin{lemma}[Garbling cannot improve posterior mean-square error]
  \label{lem:posterior-mse}
  Let $M$ have finite second moments. If $M-Z-A$ is a Markov chain, then
  \begin{equation}
    \bar{\Gamma}_Z\preceq\bar{\Gamma}_A.
    \label{eq:posterior-error-order}
  \end{equation}
\end{lemma}

The next lemma states the exact average-posterior-error form needed for an arbitrary experiment $Z$. A closely related conditional maximum-entropy inequality appears as Lemma~3.2 of \citet{zhao2025}.

\begin{lemma}[Information lower bound from posterior-error covariance]
  \label{lem:info-lower-bound}
  Let
  \[
    M\sim\mathcal{N}(\mu_M,\Sigma_M),
    \qquad
    \Sigma_M\succ0,
  \]
  and let $Z$ be any experiment about $M$, such that $\bar{\Gamma}_Z\succ0$. Then
  \begin{equation}
    I(M;Z)
    \geq
    \frac12\log_2
    \frac{\det\Sigma_M}{\det\bar{\Gamma}_Z}.
    \label{eq:posterior-error-information-bound}
  \end{equation}
\end{lemma}

\begin{proof}[Proof of \cref{lem:info-lower-bound}]
  Since $M$ is Gaussian, Equation~\eqref{eq:gaussian-entropy} gives
  \begin{align}
    I(M;Z) & = h(M)-h(M\mid Z)
    \nonumber
    \\
           & =\frac12\log_2\big((2\pi e)^{d_M}\det \Sigma_M \big) -h(M\mid Z).
    \label{eq:info-lower-bound-expansion}
  \end{align}
  If $h(M\mid Z)=-\infty$, then $I(M;Z)=+\infty$ and there is nothing to prove, so assume $h(M\mid Z)>-\infty$. For the second term in \eqref{eq:info-lower-bound-expansion}, we have
  \begin{align}
    h(M\mid Z)
     & =
    h\big(M-\mathbb{E}[M\mid Z]\,\big|\,Z\big),
    \label{eq:shift-M}
    \\
     & \leq
    h\big(M-\mathbb{E}[M\mid Z]\big),
    \label{eq:drop-conditioning}
    \\
     & \leq
    h\big( \mathcal{N}(0,\bar{\Gamma}_Z) \big),
    \label{eq:Gaussian-maximum-entropy}
    \\
     & =
    \frac12\log_2\big((2\pi e)^{d_M}\det\bar{\Gamma}_Z\big).
    \label{eq:conditional-entropy-bound}
  \end{align}
  In \eqref{eq:shift-M} we used that the conditional mean $\mathbb{E}[M\mid Z]$ is a function of $Z$, so the conditional differential entropy is invariant under this shift. In \eqref{eq:drop-conditioning} we drop the conditioning on $Z$, which can only increase the differential entropy. In \eqref{eq:Gaussian-maximum-entropy} we apply the Gaussian maximum-entropy bound, which states that among all random vectors with a given covariance matrix, the Gaussian distribution has the largest differential entropy. Here, the covariance of $M-\mathbb{E}[M\mid Z]$ is exactly $\bar{\Gamma}_Z$ by \eqref{eq:average-posterior-error}. Substituting \eqref{eq:conditional-entropy-bound} into \eqref{eq:info-lower-bound-expansion} cancels the factors $(2\pi e)^{d_M}$ and yields \eqref{eq:posterior-error-information-bound}.
\end{proof}

\subsection{Linear-Gaussian Garbling Kernels}
\label[appendix]{app:gaussian-garbling}

The next lemma shows that, whenever $\Gamma_Z \preceq \Gamma_A$, a copy of the experiment $A$ can be generated from $Z$ by an explicit linear-Gaussian kernel; it is what makes the attaining coupling in \cref{thm:gaussian-broja} jointly Gaussian.

\begin{lemma}[Explicit linear-Gaussian garbling]
  \label{lem:gaussian-garbling}
  Let $M \sim \mathcal{N}(\mu_M, \Sigma_M)$ with $\Sigma_M \succ 0$, and let
  \begin{align}
    A & = \mu_A + H_A (M - \mu_M) + N_A, & N_A & \sim \mathcal{N}(0, R_A),
    \label{eq:gaussian-experiment-A}
    \\
    Z & = \mu_Z + H_Z (M - \mu_M) + N_Z, & N_Z & \sim \mathcal{N}(0, R_Z),
  \end{align}
  with $R_A \succ 0$, $R_Z \succ 0$, and $N_A, N_Z$ independent of $M$. Let $\Gamma_A := \operatorname{Cov}(M \mid A)$ and $\Gamma_Z := \operatorname{Cov}(M \mid Z)$, write $J_A := H_A^\top R_A^{-1} H_A$ and $J_Z := H_Z^\top R_Z^{-1} H_Z$, and assume
  \begin{equation}
    \Gamma_Z \preceq \Gamma_A .
    \label{eq:garbling-hypotheses}
  \end{equation}
  Then there exist a matrix $G$ and a matrix $S \succeq 0$ such that, for $W \sim \mathcal{N}(0, S)$ independent of $(M, Z)$, the variable
  \begin{equation}
    A' := \mu_A + G (Z - \mu_Z) + W
    \label{eq:garbled-copy}
  \end{equation}
  satisfies $(M, A') \overset{d}{=} (M, A)$ and the Markov chain $M - Z - A'$. In particular, $(M, Z, A')$ is jointly Gaussian. If $J_Z \succ 0$, one can take
  \begin{equation}
    G := H_A J_Z^{-1} H_Z^\top R_Z^{-1},
    \qquad
    S := R_A - H_A J_Z^{-1} H_A^\top .
    \label{eq:garbling-G-S}
  \end{equation}
\end{lemma}

\cref{lem:gaussian-garbling} makes explicit the matrix $G$ whose existence is asserted in \citet[Lemma~5]{shang2012noisy} and used, following \citet{gerdes2015equivalence}, in the proof of \citet[Theorem~2, Eqs.~(74)--(75)]{venkatesh2022partial}.

\begin{proof}
  \textbf{Case 1: $J_Z \succ 0$.} We show that $G$ and $S$ defined in~\eqref{eq:garbling-G-S} have the stated properties.

  \textbf{Step 1: $J_A \preceq J_Z$.} By the posterior-precision identity~\eqref{eq:posterior-precision}, $\Gamma_A^{-1} = \Sigma_M^{-1} + J_A$ and $\Gamma_Z^{-1} = \Sigma_M^{-1} + J_Z$. Since $0 \prec \Gamma_Z \preceq \Gamma_A$ and matrix inversion reverses the Loewner order (\cref{subsec:gaussian-matrix-tools}), $\Gamma_A^{-1} \preceq \Gamma_Z^{-1}$, and therefore
  \begin{equation}
    J_A = \Gamma_A^{-1} - \Sigma_M^{-1}
    \preceq
    \Gamma_Z^{-1} - \Sigma_M^{-1} = J_Z .
    \label{eq:JA-JZ}
  \end{equation}

  \textbf{Step 2: $G$ reproduces the channel matrix.} By the definition of $J_Z$,
  \begin{equation}
    G H_Z = H_A J_Z^{-1} \bigl(H_Z^\top R_Z^{-1} H_Z\bigr) = H_A J_Z^{-1} J_Z = H_A .
    \label{eq:G-reproduces}
  \end{equation}

  \textbf{Step 3: $S \succeq 0$.} Consider the block matrix
  \begin{equation}
    \begin{pmatrix} R_A & H_A \\ H_A^\top & J_A \end{pmatrix}.
    \label{eq:block-A}
  \end{equation}
  Since $R_A \succ 0$ and its Schur complement with respect to $R_A$ is $J_A - H_A^\top R_A^{-1} H_A = 0 \succeq 0$, the matrix~\eqref{eq:block-A} is positive semidefinite by the Schur complement criterion~\eqref{eq:schur-criterion}. Adding the positive-semidefinite matrix $\operatorname{diag}(0, J_Z - J_A)$, see \eqref{eq:JA-JZ}, gives
  \begin{equation}
    \begin{pmatrix} R_A & H_A \\ H_A^\top & J_Z \end{pmatrix} \succeq 0 .
    \label{eq:block-Z}
  \end{equation}
  Since $J_Z \succ 0$, the criterion~\eqref{eq:schur-criterion-R} applied to~\eqref{eq:block-Z} shows that the Schur complement $R_A - H_A J_Z^{-1} H_A^\top = S$ is positive semidefinite.

  \textbf{Step 4: the law of $(M, A')$.} Substituting $Z - \mu_Z = H_Z (M - \mu_M) + N_Z$ into \eqref{eq:garbled-copy} and using \eqref{eq:G-reproduces},
  \begin{equation}
    A' = \mu_A + H_A (M - \mu_M) + N', \qquad N' := G N_Z + W .
  \end{equation}
  This has the same form as $A$ in \eqref{eq:gaussian-experiment-A}, with the same mean and channel matrix. We only need to check that the noise $N'$, which has zero mean, has the same covariance as $N_A$. Indeed, $N'$ is Gaussian and independent of $M$. Since $N_Z$ and $W$ are independent, and $J_Z$ is symmetric, its covariance is
  \begin{equation}
    \begin{aligned}
      \operatorname{Cov}(N')
       & = G\operatorname{Cov}(N_Z)G^\top + \operatorname{Cov}(W)
      \\
       & = G R_Z G^\top + S
      \\
       & = H_A J_Z^{-1} H_Z^\top R_Z^{-1}\, R_Z\, R_Z^{-1} H_Z J_Z^{-1} H_A^\top + S
      \\
       & = H_A J_Z^{-1} \bigl(H_Z^\top R_Z^{-1} H_Z\bigr) J_Z^{-1} H_A^\top + S
      \\
       & = H_A J_Z^{-1} H_A^\top + S = R_A .
    \end{aligned}
  \end{equation}
  Thus $(M, A')$ is jointly Gaussian with the same mean, channel matrix, and noise covariance as $(M, A)$, so $(M, A') \overset{d}{=} (M, A)$. The Markov chain $M - Z - A'$ holds because $A'$ is a function of $Z$ and of $W$, which is independent of $(M, Z)$.

  \textbf{Case 2: $J_Z$ singular.} Let $U := \ker J_Z$. Step~1 does not use the invertibility of $J_Z$, so \eqref{eq:JA-JZ} holds, and for every $u \in U$,
  \begin{equation}
    0 \leq u^\top J_A u \leq u^\top J_Z u = 0 .
  \end{equation}
  Since $u^\top J_A u = \|R_A^{-1/2} H_A u\|^2$ and $u^\top J_Z u = \|R_Z^{-1/2} H_Z u\|^2$, this gives
  \begin{equation}
    H_A u = 0, \qquad H_Z u = 0 \qquad \text{for all } u \in U .
    \label{eq:H-vanish-on-U}
  \end{equation}
  Let $\Pi$ be the projector with kernel $U$ that is orthogonal with respect to the inner product $\langle u, v \rangle := u^\top \Sigma_M^{-1} v$, let $k := d_M - \dim U$, and let $B \in \mathbb{R}^{d_M \times k}$ have columns forming a basis of $\operatorname{range} \Pi$. Split
  \begin{equation}
    \begin{aligned}
      M - \mu_M & = M_1 + M_2,                                                      \\
      M_1       & := \Pi (M - \mu_M) = B m, \qquad m := (B^\top B)^{-1} B^\top M_1, \\
      M_2       & := (I - \Pi)(M - \mu_M) \in U .
    \end{aligned}
  \end{equation}
  The vectors $M_1$ and $M_2$ are jointly Gaussian with $\operatorname{Cov}(M_1, M_2) = \Pi \Sigma_M (I - \Pi)^\top = 0$, hence independent, and $m \in \mathbb{R}^k$ is Gaussian with positive-definite covariance, since $\operatorname{Cov}(M_1) = \Pi \Sigma_M \Pi^\top$ has rank $k$. By \eqref{eq:H-vanish-on-U}, $H_A M_2 = H_Z M_2 = 0$, so
  \begin{equation}
    A = \mu_A + (H_A B)\, m + N_A,
    \qquad
    Z = \mu_Z + (H_Z B)\, m + N_Z .
    \label{eq:reduced-channels}
  \end{equation}
  Thus $(m, A, Z)$ is an instance of the setting of the lemma with target $m$, channel matrices $H_A B$ and $H_Z B$, and information matrices $B^\top J_A B$ and $B^\top J_Z B$. Steps~2--4 of Case~1 use the hypothesis only through \eqref{eq:JA-JZ} and $J_Z \succ 0$, and both hold for the reduced instance: $B^\top J_A B \preceq B^\top J_Z B$ by congruence, and $B^\top J_Z B \succ 0$ because $Bc \in \operatorname{range}\Pi \setminus \{0\}$ is not in $U = \ker J_Z$ for $c \neq 0$. Case~1 applied to the reduced instance therefore yields, with
  \begin{equation}
    G := (H_A B)(B^\top J_Z B)^{-1} (H_Z B)^\top R_Z^{-1},
    \qquad
    S := R_A - (H_A B)(B^\top J_Z B)^{-1} (H_A B)^\top \succeq 0,
    \label{eq:reduced-G-S}
  \end{equation}
  a variable $A' := \mu_A + G(Z - \mu_Z) + W$, $W \sim \mathcal{N}(0, S)$ independent of $(M, Z)$, such that $(m, A') \overset{d}{=} (m, A)$ and $m - Z - A'$. The Markov chain $M - Z - A'$ holds because $A'$ is a function of $Z$ and $W$. Finally, $M = \mu_M + Bm + M_2$, where $M_2$ is independent of $(m, A)$ by \eqref{eq:reduced-channels} and of $(m, A')$ because $A'$ depends on $(m, N_Z, W)$ only; together with $(m, A') \overset{d}{=} (m, A)$ this gives $(M, A') \overset{d}{=} (M, A)$.
\end{proof}

\begin{remark}[Two-stage interpretation]
  The kernel~\eqref{eq:garbling-G-S} factors through the sufficient statistic $T_Z := H_Z^\top R_Z^{-1}(Z - \mu_Z)$, which satisfies $T_Z \sim \mathcal{N}(J_Z (M-\mu_M), J_Z)$ conditionally on $M$: since $G(Z - \mu_Z) = H_A J_Z^{-1} T_Z$, the copy $A'$ is obtained by first forming from $Z$ the unbiased estimate $J_Z^{-1} T_Z$ of $M - \mu_M$, whose conditional covariance is $J_Z^{-1}$, and then passing it through the channel of $A$ with the reduced noise covariance $S = R_A - H_A J_Z^{-1} H_A^\top$.
\end{remark}

\newpage
\section{PROOFS FOR SECTION~\closedformsecref}
\label[appendix]{sec:supp-section3}


\subsection{Proof of \texorpdfstring{\cref{prop:broja-lower-bound}}{Proposition \ref{prop:broja-lower-bound}}}
\label[appendix]{supp:proof-broja-lower-bound}
This section proves \cref{prop:broja-lower-bound} from
\cref{sec:closed-form-broja}.

We first restate the proposition.

\noindent\textbf{\cref{prop:broja-lower-bound} (BROJA union information lower bound).}
Under \cref{ass:gaussian-marginals},
\[
  I_{\cup}^{\BROJA}(M;X,Y)
  \geq
  \frac12\log_2
  \frac{\det\Sigma_M}{\det\Gamma^\star}.
\]

\begin{proof}
  Fix $Q\in\Delta(P)$ and regard the pair of sources $(X,Y)$, with joint law $Q$, as a single experiment about $M$. Denote by $\bar{\Gamma}^{Q}_{XY}$ its average posterior-error covariance, as defined in Equation~\eqref{eq:average-posterior-error}.

  Since $X$ is a coordinate of $(X,Y)$, the Markov chain
  \[
    M-(X,Y)-X
  \]
  holds under $Q$. Moreover, $Q_{MX}=P_{MX}$ by~\cref{def:compatible-couplings}, and $(M,X)$ is jointly Gaussian under $P_{MX}$ by \cref{ass:gaussian-marginals}, so the average posterior-error covariance of $X$ under $Q$ equals $\Gamma_X$. \cref{lem:posterior-mse} therefore gives
  \[
    \bar{\Gamma}^{Q}_{XY}\preceq\Gamma_X.
  \]
  The same argument with the roles of $X$ and $Y$ interchanged gives
  \[
    \bar{\Gamma}^{Q}_{XY}\preceq\Gamma_Y.
  \]

  If $I_Q(M;X,Y)=+\infty$, the claimed inequality is immediate. Otherwise $h(M\mid X,Y)>-\infty$, which forces $\bar{\Gamma}^{Q}_{XY}\succ0$: if $\bar{\Gamma}^{Q}_{XY}$ were singular, the error vector $M-\mathbb{E}_Q[M\mid X,Y]$ would be supported on a proper subspace, so that $h(M\mid X,Y)\leq h(M-\mathbb{E}_Q[M\mid X,Y])=-\infty$. Hence $\bar{\Gamma}^{Q}_{XY}\in\mathcal{F}_{X,Y}$, the feasible set in Equation~\eqref{eq:feasible-posterior-covariances}, and the optimality of $\Gamma^\star$ established in \cref{prop:max-det} yields
  \[
    \det\bar{\Gamma}^{Q}_{XY}\leq\det\Gamma^\star.
  \]
  Applying \cref{lem:info-lower-bound} to the experiment $(X,Y)$ and then this determinant inequality,
  \[
    I_Q(M;X,Y)
    \geq
    \frac12\log_2
    \frac{\det\Sigma_M}{\det\bar{\Gamma}^{Q}_{XY}}
    \geq
    \frac12\log_2
    \frac{\det\Sigma_M}{\det\Gamma^\star}.
  \]
  Since $Q\in\Delta(P)$ was arbitrary, taking the infimum over $\Delta(P)$ in the definition~\eqref{eq:broja-union} of the BROJA union information proves Equation~\eqref{eq:broja-lower-bound}.
\end{proof}

\subsection{Closed-form redundancy and synergy for a jointly Gaussian triple}
\label[appendix]{supp:proof-full-pid}

Assume, in addition to \cref{ass:gaussian-marginals}, that the observed triple $(M,X,Y)$ is nondegenerate jointly Gaussian, and define
\[
  \Gamma_{XY}:=\Cov(M\mid X,Y).
\]
Let $\RI^{\BROJA}$ and $\SI^{\BROJA}$ denote the BROJA redundancy and synergy, respectively. Then
\[
  \RI^{\BROJA}
  =
  \frac12\log_2
  \frac{\det\Sigma_M\,\det\Gamma^\star}
  {\det\Gamma_X\,\det\Gamma_Y},
\]
and
\[
  \SI^{\BROJA}
  =
  \frac12\log_2
  \frac{\det\Gamma^\star}
  {\det\Gamma_{XY}}.
\]
Together with the two unique-information formulas in \cref{thm:gaussian-broja}, this gives all four BROJA PID components in closed form.

\begin{proof}
  The BROJA redundancy satisfies
  \[
    \RI^{\BROJA}
    =
    I(M;X)-\UI_X^{\BROJA}.
  \]
  By Equation~\eqref{eq:ui-union-relation} and \cref{thm:gaussian-broja},
  \[
    \UI_X^{\BROJA}
    =
    I_{\cup}^{\BROJA}(M;X,Y)-I(M;Y),
  \]
  so
  \[
    \RI^{\BROJA}
    =
    I(M;X)+I(M;Y)-I_{\cup}^{\BROJA}(M;X,Y).
  \]
  Since $(M,X)$ and $(M,Y)$ are jointly Gaussian,
  \[
    I(M;X)
    =
    \frac12\log_2
    \frac{\det\Sigma_M}{\det\Gamma_X},
    \qquad
    I(M;Y)
    =
    \frac12\log_2
    \frac{\det\Sigma_M}{\det\Gamma_Y}.
  \]
  \cref{thm:gaussian-broja} gives
  \[
    I_{\cup}^{\BROJA}(M;X,Y)
    =
    \frac12\log_2
    \frac{\det\Sigma_M}{\det\Gamma^\star}.
  \]
  Substituting these three expressions yields
  \[
    \RI^{\BROJA}
    =
    \frac12\log_2
    \frac{\det\Sigma_M\,\det\Gamma^\star}
    {\det\Gamma_X\,\det\Gamma_Y}.
  \]

  For synergy, the BROJA PID identity gives
  \[
    \SI^{\BROJA}
    =
    I(M;X,Y)-I_{\cup}^{\BROJA}(M;X,Y).
  \]
  Because the observed triple $(M,X,Y)$ is jointly Gaussian,
  \[
    I(M;X,Y)
    =
    \frac12\log_2
    \frac{\det\Sigma_M}{\det\Gamma_{XY}}.
  \]
  Subtracting the union-information formula from \cref{thm:gaussian-broja} therefore gives
  \[
    \SI^{\BROJA}
    =
    \frac12\log_2
    \frac{\det\Gamma^\star}{\det\Gamma_{XY}},
  \]
  which proves the result.
\end{proof}

\clearpage
\section{NUMERICAL EXPERIMENTS AND COMPUTATIONAL COST}
\label[appendix]{sec:supp-numerical}
\label[appendix]{sec:computing-infrastructure}

\paragraph{Computing infrastructure.}
Experiments were run on an internal SLURM-managed computing cluster. Within each compute job, PID computations were performed serially. Each CPU job was allocated 8 CPU cores and 8--12 GB of host memory. Each GPU job was allocated one NVIDIA RTX 6000 Ada Generation GPU with 48 GB of GPU memory, together with 8 CPU cores and 24 GB of host memory. Runtime comparisons were executed on cluster nodes equipped with Intel Xeon E5-2680 v2/v3/v4 or Intel Xeon E5-2660 v4 processors, using eight CPU threads per job.

\paragraph{Computational cost}

Given the covariance matrices, evaluating Theorem~1 requires the two conditional covariances $\Gamma_X$ and $\Gamma_Y$, obtained from the Schur complements~\eqref{eq:posterior-schur} at cost $O(d_X^3 + d_M d_X^2 + d_M^2 d_X)$ and $O(d_Y^3 + d_M d_Y^2 + d_M^2 d_Y)$, followed by a single symmetric-definite generalized eigenproblem of size $d_M$, at cost $O(d_M^3)$. The total cost is therefore $O(d_X^3 + d_Y^3 + d_M^3)$, with no iterations, initialization, or stopping criteria. The redundancy in Corollary~1 requires only $\Sigma_M$, $\Gamma_X$, $\Gamma_Y$, and $\Gamma^\star$, which are already available. The synergy additionally requires $\Gamma_{XY} = \operatorname{Cov}(M \mid X,Y)$. Writing $s=d_X+d_Y$, computing $\Gamma_{XY}$ via its Schur complement costs $O(s^3 + d_M s^2 + d_M^2 s)$, while computing its log-determinant costs $O(d_M^3)$. Thus, the overall cubic-order bound $O(d_X^3 + d_Y^3 + d_M^3)$ remains unchanged. By contrast, the existing gradient-based methods repeatedly perform comparable linear-algebra operations throughout an iterative optimization procedure.

\paragraph{Covariance set-up}
In our numerical experiments, we used population covariances, so no finite-sample bias correction was required. In each trial all methods were evaluated on the same generated covariance.
\subsection{Runtime comparisons}
\label[appendix]{subsec:supp-runtime-comparisons}

We compared Eigen-PID, Thin-PID, and Tilde-PID using joint population covariance matrices of a Gaussian target \(M\) and two sources \(X,Y\). We used the model
\[
  M\sim\mathcal{N}(0,I_{d_M}),\qquad
  X=H_XM+\varepsilon_X,\qquad
  Y=H_YM+\varepsilon_Y.
\]
The channel matrices \(H_X\in\mathbb{R}^{d_X\times d_M}\) and \(H_Y\in\mathbb{R}^{d_Y\times d_M}\) had independent standard normal entries. The noise vectors \(\varepsilon_X\sim\mathcal{N}(0,I_{d_X})\) and \(\varepsilon_Y\sim\mathcal{N}(0,I_{d_Y})\) were independent of each other and of \(M\).

For each value of \(d_X\) within each configuration, we generated five joint population covariance matrices of \((M,X,Y)\) using different random seeds and evaluated all three methods on the same matrices.

All three methods were run on CPU using the same computational resources. We timed each complete PID computation, excluding construction of the input covariance matrix. Thin-PID and Tilde-PID used their default optimization settings.

\Cref{fig:runtime_speedup} shows the runtime comparison with \(d_M=1000\) and \(d_Y=1000\) while varying \(d_X\). \Cref{fig:runtime_grid}A extends this comparison to additional fixed values of \(d_M\) and \(d_Y\). In \Cref{fig:runtime_grid}B, we vary the dimension of \(M\) while fixing the dimensions of \(X\) and \(Y\).

\begin{figure}[!htbp]
  \centering
  \includegraphics{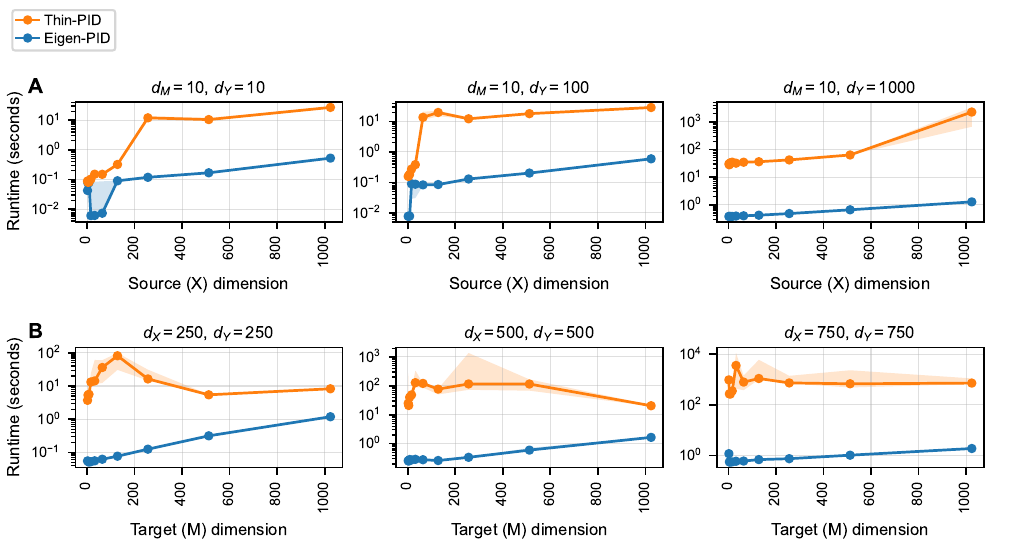}
  \caption{\textbf{Runtime comparison across target and source dimensions.} Runtime of Eigen-PID and Thin-PID under the same random Gaussian channel construction used in \cref{fig:runtime_speedup}. \textbf{(A)} Varying $d_X$ for fixed combinations of $d_M$ and $d_Y$. \textbf{(B)} Varying $d_M$ for fixed, equal source dimensions $d_X=d_Y$. Both methods were evaluated on the same randomly generated problem instances at each dimensional configuration. Lines show median runtimes across instances; shaded regions indicate the 25th--75th percentile range. Runtime is shown on a logarithmic scale. Eigen-PID remains substantially faster than Thin-PID across both sweeps, with the runtime gap generally increasing at larger dimensions. We did not include Tilde-PID in these extended sweeps because it was generally slower than Thin-PID in \cref{fig:runtime_speedup}, consistent with the findings of \citet{zhao2025}.}
  \label{fig:runtime_grid}
\end{figure}

\FloatBarrier
\clearpage

\subsection{Gain sweep}
\label[appendix]{subsec:supp-benchmark-examples}
We reproduced the gain-sweep benchmark from Example~8 of \citet{venkatesh2023}, using the experimental settings of \citet{zhao2025}.

Thin-PID and Tilde-PID compute \(I(M;(X,Y))\) using the regularized source covariance
\[
  \Sigma_{(X,Y)} + 10^{-7}I.
\]
To compare PID values under the same numerical convention, we applied this regularization to Eigen-PID as well in the gain-sweep comparison. For Tilde-PID, we evaluated the final union information without regularization, while retaining regularization during optimization.

We consider a two-dimensional Gaussian target
\[
  M=(M_1,M_2)^\top,
  \qquad
  M_1,M_2\sim\mathcal{N}(0,1),
\]
with \(M_1\) and \(M_2\) independent. The sources are defined by
\[
  X_1=\alpha M_1+\varepsilon_{X,1},
  \qquad
  Y_1=M_1+\varepsilon_{Y,1},
\]
where $\alpha$ is the gain value and
\[
  X_2=M_2+\varepsilon_{X,2},
  \qquad
  Y_2=\sqrt{2}\,M_2+\varepsilon_{Y,2},
\]
where all noise variables are independent standard Gaussians and independent of \(M\). Since the two triples \((M_1,X_1,Y_1)\) and \((M_2,X_2,Y_2)\) are independent, the BROJA PID is additive across the corresponding scalar subsystems \citep{bertschinger2014quantifying}. Their unique-information contributions can therefore be summed.

In matrix form,
\[
  X=H_XM+\varepsilon_X,
  \qquad
  Y=H_YM+\varepsilon_Y,
\]
with
\[
  H_X=\operatorname{diag}(\alpha,1),
  \qquad
  H_Y=\operatorname{diag}(1,\sqrt{2}).
\]

Varying \(\alpha\) changes the posterior uncertainty about \(M_1\) after observing \(X_1\):
\[
  \operatorname{Var}(M_1\mid X_1)=\frac{1}{1+\alpha^2},
  \qquad
  \operatorname{Var}(M_1\mid Y_1)=\frac{1}{2}.
\]
For \(\alpha<1\), \(X_1\) leaves greater posterior uncertainty and is therefore less informative than \(Y_1\); at \(\alpha=1\), the two sources are equally informative; and for \(\alpha>1\), \(X_1\) leaves less posterior uncertainty and becomes more informative. Thus, the source carrying unique information about \(M_1\) switches from \(Y\) to \(X\) at \(\alpha=1\). Along \(M_2\), \(Y_2\) remains more informative than \(X_2\). The sweep therefore crosses from a regime in which \(Y\) dominates both target directions to one in which each source is more informative along a different direction.

\FloatBarrier
\clearpage

\subsection{Initializing numerical estimators at the Eigen-PID optimum}
\label[appendix]{sec:optimal-covariance-initialization}
We initialized Thin-PID and Tilde-PID at the optimal Gaussian coupling constructed by Eigen-PID and ran both optimizers with their default settings. We then compared their final union-information values, evaluated without regularization, with the Eigen-PID value. Neither method produced a systematic reduction beyond numerical precision (\cref{fig:optimal-initialization}).

\begin{figure}[!htbp]
  \centering
  \includegraphics{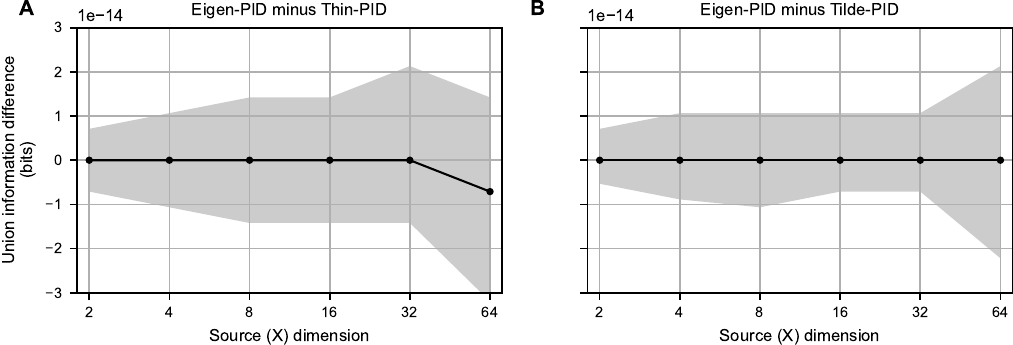}
  \caption{\textbf{Numerical validation of the Eigen-PID optimum.} Union-information differences after initializing Thin-PID and Tilde-PID at the Eigen-PID optimum: \textbf{(A)} $I_{\cup}^{\mathrm{Eigen}}-I_{\cup}^{\mathrm{Thin}}$; \textbf{(B)} $I_{\cup}^{\mathrm{Eigen}}-I_{\cup}^{\mathrm{Tilde}}$. The dimensional sweep follows the setup described in \autoref{fig:runtime_speedup}, with $d_M=d_Y=10$ and $d_X$ varying up to 64. Markers show medians across random problem instances; shaded regions show the 25th--75th percentile range.}

  \label{fig:optimal-initialization}
\end{figure}

Below, we explain how we obtained the optimal initialization from the Eigen-PID solution.

Let $\Sigma$ denote the covariance of $(M,X,Y)$, ordered as $[M,X,Y]$. Because the BROJA optimization preserves the two target--source marginals, within the jointly Gaussian parameterization the only free covariance block is $\Sigma_{XY}$. From the fixed marginals we compute
\begin{equation}
  \Gamma_X = \Sigma_{M\mid X},
  \qquad
  \Gamma_Y = \Sigma_{M\mid Y},
\end{equation}
and obtain the Eigen-PID optimizer $\Gamma^\star$ from the generalized eigenproblem
\begin{equation}
  \Gamma_Y v_i = \lambda_i \Gamma_X v_i,
  \qquad
  V^\top \Gamma_X V = I,
\end{equation}
as
\begin{equation}
  \Gamma^\star
  =
  \Gamma_X V
  \operatorname{diag}\!\bigl(\min\{1,\lambda_i\}\bigr)
  V^\top \Gamma_X .
\end{equation}

To initialize the numerical methods at an attaining Gaussian coupling, define
\begin{equation}
  C^\star := \Sigma_M-\Gamma^\star .
\end{equation}
One compatible optimal source cross-covariance is
\begin{equation}
  \Sigma_{XY}^\star
  =
  \Sigma_{XM}(C^\star)^\dagger\Sigma_{MY},
  \label{eq:sigma-xy-star}
\end{equation}
where ${}^\dagger$ denotes the Moore--Penrose pseudoinverse. Replacing only the $(X,Y)$ and $(Y,X)$ blocks of $\Sigma$ by $\Sigma_{XY}^\star$ and $\Sigma_{XY}^{\star\top}$ yields a covariance $\Sigma^\star$ that preserves both target--source marginals and satisfies
\begin{equation}
  \Sigma_{M\mid X,Y}^\star=\Gamma^\star .
\end{equation}

Both numerical estimators operate on the cross-covariance of the source noises after conditioning on and whitening with respect to $M$. Writing
\begin{align}
  N_X & = \Sigma_{X\mid M}, &
  N_Y & = \Sigma_{Y\mid M},
\end{align}
the optimal initialization is therefore
\begin{equation}
  S^\star
  =
  N_X^{-1/2}
  \left(
  \Sigma_{XY}^\star
  -
  \Sigma_{XM}\Sigma_M^{-1}\Sigma_{MY}
  \right)
  N_Y^{-1/2}.
  \label{eq:optimal-whitened-initialization}
\end{equation}
$S^\star$ is created by inputting $\Sigma_{XY}^\star$ to Thin-PID or Tilde-PID and then using the built-in whitening function of Thin-PID or Tilde-PID. We use $S^\star$ in place of the default initialization of Thin-PID and Tilde-PID, while leaving their subsequent optimization procedures unchanged. In Tilde-PID, to compare the corresponding quantities, the regularization term (\(10^{-7}\)) was used during the optimization but omitted when evaluating the final union information. Thus, this experiment tests whether either numerical optimizer reduces union information below the Eigen-PID value.

\subsection{Software and Software Licenses}
\label[appendix]{sec:license}
\subsubsection{Eigen-PID}
The Eigen-PID Python package and code for reproducing our numerical experiments are available at \url{https://github.com/brainsandmachines/Eigen_PID}. Eigen-PID uses PyTorch for CPU and GPU computations.

\subsubsection{Baseline Implementations and Licenses}
We modified both baseline implementations to allow initialization at the Eigen-PID optimum for the experiment in \cref{sec:optimal-covariance-initialization}. For Tilde-PID, we also removed regularization from the final union-information evaluation, without changing the optimization procedure.

For Tilde-PID, we used the implementation released by Venkatesh et al. at \url{https://github.com/praveenv253/gpid}. The upstream package declares an MIT license.

For Thin-PID, we used the implementation released by Zhao et al. at \url{https://github.com/warrenzha/flow-pid}, based on commit \texttt{4a30b1a}. At the time of access, the public repository did not specify a software license.

\section{EIGEN-PID APPLIED TO VISUAL FMRI RESPONSES}
\label[appendix]{sec:fmri}

\begin{figure}[h]
  \centering
  \includegraphics{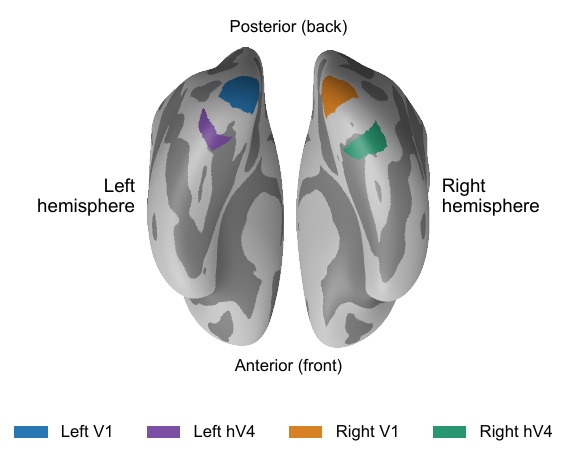}
  \caption{\textbf{Source and target regions.} Participant 1's V1 and hV4 regions are shown on inflated standard cortical surfaces (fsaverage), viewed from below and behind. Light and dark gray indicate cortical folds (gyri) and grooves (sulci). Both the left hemisphere and right hemisphere V1 regions are used as sources, with one of the hV4 regions (left or right) as the target.}
  \label{fig:fmri_cortex}
\end{figure}
Resampling-based bias correction of BROJA PID \citep{lorenz2026} requires repeated PID estimates, which are computationally expensive for high-dimensional data.  Here, we demonstrate the practical utility of Eigen-PID in substantially reducing computation time for this procedure using human functional magnetic resonance imaging (fMRI) responses to natural images.

\paragraph{Experimental setup.} We analyzed fMRI responses from the Natural Scenes Dataset (NSD), which contains fMRI measurements collected from eight participants as they viewed natural images \citepfmri{allen2022}. We used the Algonauts 2023 challenge training data \citepfmri{gifford2023}, consisting of preprocessed NSD responses at selected visually responsive cortical vertices. Each vertex is a sampling location on a reconstructed cortical surface. In the Algonauts data, responses were standardized within each scan session and averaged across image repetitions. The training data contained 8,779--9,841 distinct images and associated cortical responses per participant.

We selected primary visual cortex (V1) and human visual area V4 (hV4) in each hemisphere using the NSD-provided labels (\autoref{fig:fmri_cortex}), combining the dorsal and ventral subdivisions of V1. We removed vertices with identical response profiles across images within each region. For each image, we combined the responses across vertices within each region into a single vector. We performed two PID analyses per participant, using left and right V1 as sources in both, with left hV4 as the target in one analysis and right hV4 in the other. The three vectors for the same image formed one joint observation $(M,X,Y)$: $X$ and $Y$ represent left and right V1, and $M$ represents the chosen hV4 target. The image count $N$ and region dimensions $d_X$, $d_Y$, and $d_M$ varied across participants (\cref{tab:fmri_dimensions}), with a joint dimension $p=d_X+d_Y+d_M$ of 2,813--3,704.

\FloatBarrier

\begin{table}[!htb]
  \centering
  \caption{\textbf{Image counts and region dimensions.} Dimensions indicate the number of cortical vertices included in each region for each participant. Left and right V1 served as sources in all analyses; targets were either left hV4 or right hV4.}
  \label{tab:fmri_dimensions}
  \begin{tabular}{lrrrrrrrr}
    \toprule
    Participant                  & S1    & S2    & S3    & S4    & S5    & S6    & S7    & S8    \\
    \midrule
    Images ($N$)                 & 9,841 & 9,841 & 9,082 & 8,779 & 9,841 & 9,082 & 9,841 & 8,779 \\
    Left V1 vertices ($d_{X}$)   & 1,520 & 1,438 & 1,375 & 1,116 & 1,633 & 1,237 & 1,402 & 1,011 \\
    Right V1 vertices ($d_{Y}$)  & 1,426 & 1,287 & 1,293 & 1,159 & 1,266 & 1,251 & 1,142 & 1,384 \\
    Left hV4 vertices ($d_{M}$)  & 527   & 590   & 450   & 538   & 579   & 434   & 522   & 531   \\
    Right hV4 vertices ($d_{M}$) & 758   & 641   & 427   & 621   & 334   & 376   & 377   & 530   \\
    \bottomrule
  \end{tabular}
\end{table}

\FloatBarrier

\paragraph{Debiased PID estimates for real fMRI data.} For each participant and target, we estimated the joint sample covariance without shrinkage or dimensionality reduction and used Eigen-PID to compute the BROJA components of the fitted joint Gaussian: unique information $\operatorname{UI}_X$ and $\operatorname{UI}_Y$, redundancy $\operatorname{RI}$, and synergy $\operatorname{SI}$. We applied the merged correction of \citet{lorenz2026}, using our reimplementation (\texttt{debias\_eigen\_pid.py}), which uses Eigen-PID in place of Tilde-PID. We corrected the estimates of $I(M;X)$, $I(M;Y)$, and $I(M;X,Y)$ by subtracting their analytical estimation biases under independent, identically distributed Gaussian sampling. For union information $I_{\cup}^{\mathrm{BROJA}}(M;X,Y)$, we averaged two bias estimates:

\begin{enumerate}
  \item \textbf{Permutation estimate:} the mean union information across 20 random permutations of the target observations.
  \item \textbf{Parametric bootstrap estimate:} the mean union information estimated from 20 datasets sampled from the fitted Gaussian distribution, minus the union information computed from the original data without resampling.
\end{enumerate}

For each permuted or bootstrap dataset, we re-estimated the sample covariance and computed union information using Eigen-PID. We used 20 resamples of each type to match the implementation of \citet{lorenz2026}. After subtracting the averaged bias estimate from the uncorrected union-information estimate, we reconstructed the four PID components using the corrected information quantities, retaining negative estimates.

Redundancy was the largest mean debiased component for both targets (51.4 and 52.0 bits; \mbox{\autoref{fig:fmri_debiased}A--B}). When the target was left hV4, left V1 contributed more unique information than right V1 in seven of eight participants. When the target was right hV4, right V1 contributed more unique information than left V1 in all eight participants.

\begin{figure}[bthp]
  \centering
  \includegraphics{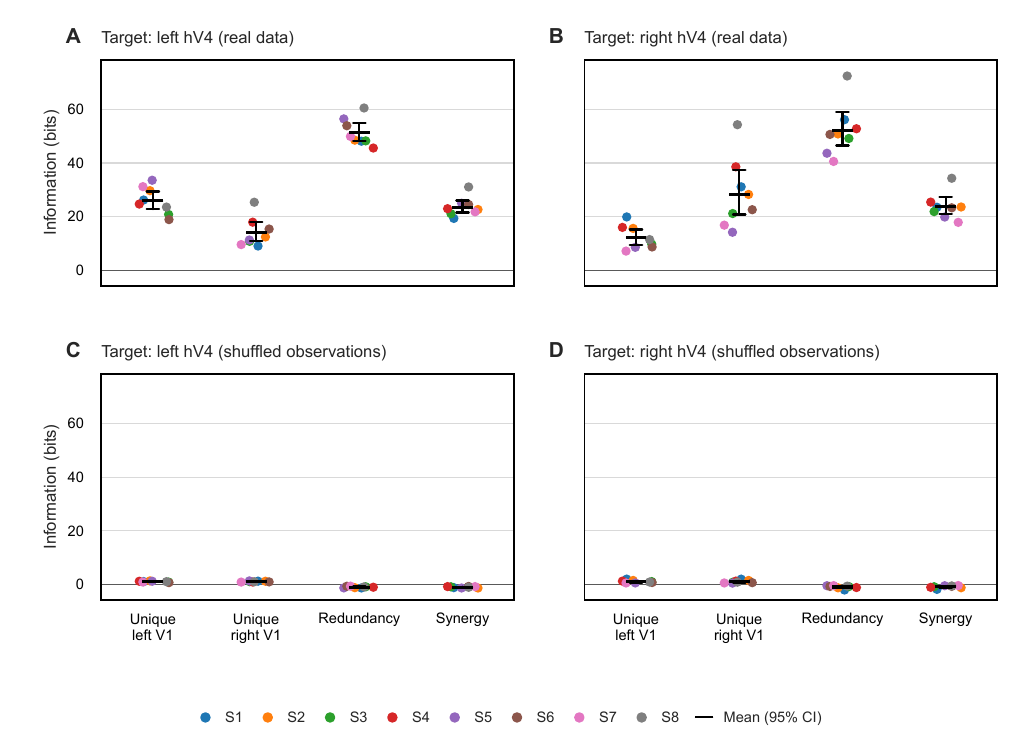}
  \caption{\textbf{Bias-corrected PID estimates.} All panels use the same scale in bits. Colors identify participants; black marks show means and 95\% percentile confidence intervals obtained by resampling the eight participants 10,000 times, independently of the bias-correction procedure. \textbf{(A)} Left hV4, bias-corrected PID of real fMRI responses. \textbf{(B)} Right hV4, bias-corrected PID of real fMRI responses. \textbf{(C)} Left hV4, bias-corrected PID of shuffled fMRI responses. \textbf{(D)} Right hV4, bias-corrected PID of shuffled fMRI responses. Bias correction preserves substantial information in real responses while leaving small residual component biases in shuffled controls.}
  \label{fig:fmri_debiased}
\end{figure}

\paragraph{Debiased PID estimates for shuffled controls.} To assess residual bias in the absence of target--source dependence, we created one shuffled-target control dataset per participant and target. Shuffling preserved the paired sources and within-region response structure while removing target--source correspondence. Applying the same bias correction as for the real data brought all four components close to zero (absolute mean estimates across participants ranged from 0.90 to 1.05 bits across components and targets; \mbox{\autoref{fig:fmri_debiased}C--D}). The residual component biases arose from undercorrection of union information by the parametric bootstrap. With permutation-only union correction (and the same analytical mutual-information correction), all mean components were within 0.11 bits of zero.

\FloatBarrier

\paragraph{PID estimates without bias correction.} We next examined the PID estimates without bias correction to assess the magnitude and distribution of finite-sample bias. In the real data, mean joint information was 242.2 bits for left hV4 and 240.7 bits for right hV4, compared with 115.1 and 116.3 bits after correction. Mean synergy was 73.4 and 72.2 bits, respectively (\mbox{\autoref{fig:fmri_raw}A--B}), compared with 23.6 and 23.7 bits after correction (\mbox{\autoref{fig:fmri_debiased}A--B}). Bias correction therefore changed both total information and its allocation among components. In the shuffled controls, all four uncorrected components were positive, with mean joint information of 127.2 bits for left hV4 and 124.6 bits for right hV4. Finite-sample bias was particularly pronounced for synergy, the largest mean component for both targets (\mbox{\autoref{fig:fmri_raw}C--D}).

\begin{figure}[!ht]
  \centering
  \includegraphics{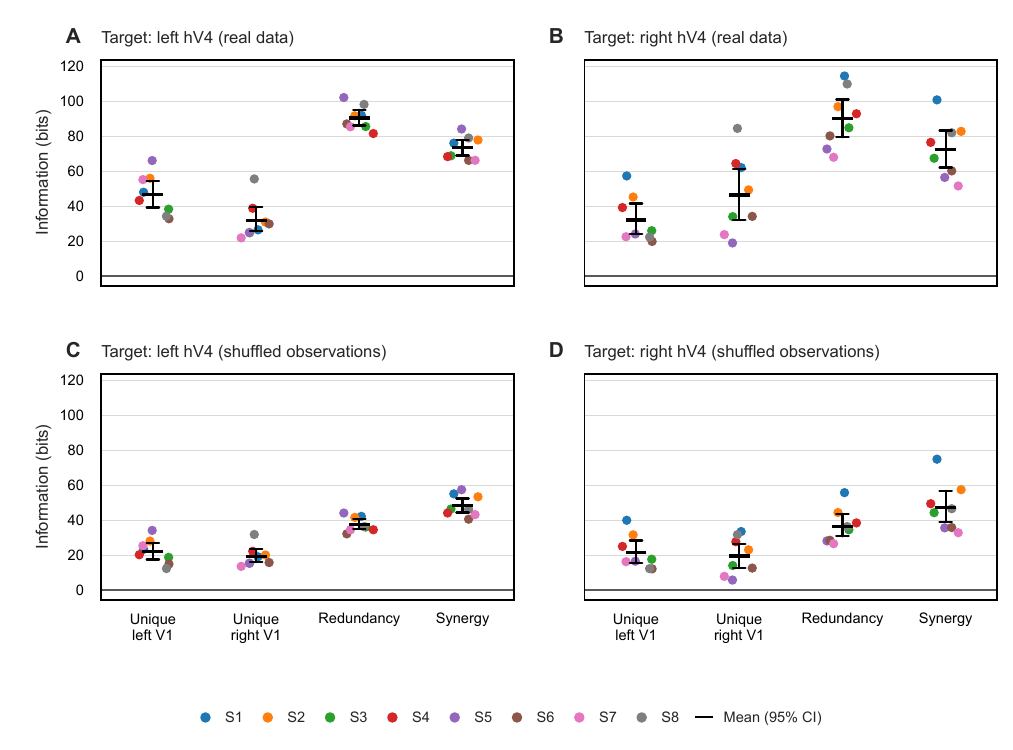}
  \caption{\textbf{Uncorrected PID estimates.} Colors and summary marks follow \autoref{fig:fmri_debiased}; all panels use a common bit scale. \textbf{(A)} Left hV4, uncorrected PID of real fMRI responses. \textbf{(B)} Right hV4, uncorrected PID of real fMRI responses. \textbf{(C)} Left hV4, uncorrected PID of shuffled fMRI responses. \textbf{(D)} Right hV4, uncorrected PID of shuffled fMRI responses. Shuffled controls reveal substantial finite-sample bias, particularly for synergy.}
  \label{fig:fmri_raw}
\end{figure}

\FloatBarrier

\paragraph{Scientific interpretation.} Our aim is to demonstrate the computational feasibility of the analysis rather than to draw neuroscientific conclusions. Such conclusions would require controlling for differences in dimensionality between regions; assessing departures from Gaussianity; and validating the bias correction at these sample sizes and dimensionalities. Furthermore, the analysis is observational and based on functional coactivation; assigning source and target roles does not establish causal direction or causal relationships. Nevertheless, the greater unique information contributed by V1 in the target's hemisphere is qualitatively consistent with the anatomical organization of primate visual pathways \citepfmri{ungerleider2008} and functional coactivation patterns within and across hemispheres \citepfmri{arcaro2015}.

\paragraph{Computational cost of bias correction.} A complete analysis across both targets and all eight participants (reproducing \mbox{\autoref{fig:fmri_debiased}A--B}) required 656 PID evaluations: one original estimate, 20 permutations, and 20 parametric bootstraps per participant and target. We ran the analyses on a workstation with a 20-core Intel Core Ultra 7 265K processor and 128\,GiB of RAM, using CPU float64 computation with 20 threads. During runtime measurements, we turned off optional numerical diagnostics, which validate the numerical conditions of the inputs without altering the PID computation. \Cref{tab:fmri_runtime} reports the computation times.

\begin{table}[!ht]
  \centering
  \caption{\textbf{Measured runtimes and projected correction costs.} Single-call times exclude covariance fitting and represent averages across all eight participants for both methods. Complete-analysis times cover both estimation and bias correction for all eight participants per target.}
  \label{tab:fmri_runtime}
  \begin{threeparttable}
    \begin{tabular}{lrrrr}
      \toprule
      \multirow{2}{*}{\textbf{Target}} & \multicolumn{2}{c}{\textbf{Single call} (seconds)} & \multicolumn{2}{c}{\textbf{Complete analysis} (minutes)}                                       \\
                                       & Eigen-PID                                          & Thin-PID                                                 & Eigen-PID & Thin-PID                \\
      \midrule
      Left hV4                         & 0.92                                               & 2,645.38                                                 & 7.20      & 14,461.43{\tnote{\dag}} \\
      Right hV4                        & 0.29                                               & 2,835.40                                                 & 7.29      & 15,500.20{\tnote{\dag}} \\
      \bottomrule
    \end{tabular}
    \begin{tablenotes}[flushleft]
      \footnotesize
      \item[\dag] Conservative projections obtained by multiplying the mean single-call time across all eight participants by 328, excluding the time to generate resampled datasets and estimate their covariances.
    \end{tablenotes}
  \end{threeparttable}
\end{table}

\FloatBarrier

Eigen-PID completed the analysis for both targets across all eight participants in 14.5 minutes, compared with a projected completion time of approximately 21 days for Thin-PID (extrapolated from one timed PID evaluation per participant and target). Moving the Eigen-PID computation from the CPU to an NVIDIA GeForce RTX 5090 GPU reduced the total runtime from 14.5 to 3.9 minutes. By eliminating iterative optimization of the PID objective, Eigen-PID makes resampling-based bias correction practical even for high-dimensional data.

\FloatBarrier

\bibliographystylefmri{plainnat-doifirst}
\bibliographyfmri{references}

\end{document}